\documentclass[journal,twoside,web]{ieeecolor}
\usepackage{generic}
\usepackage{cite}

\usepackage{amsmath,amssymb,amsfonts,amsthm}
\usepackage{algorithmicx}
\usepackage{graphicx}
\usepackage{algorithm}
\usepackage{hyperref}
\hypersetup{hidelinks=true}
\usepackage{textcomp}
\usepackage{subfig}
\usepackage{colortbl,booktabs}
\newtheoremstyle{boldstyle} 
{\topsep}   
{\topsep}   
{\normalfont}  
{}          
{\bfseries} 
{.}         
{ }         
{}          

\theoremstyle{boldstyle}

\newtheorem{definition}{Definition}
\newtheorem{assumption}{Assumption}
\newtheorem{theorem}{Theorem}
\newtheorem{lemma}{Lemma}
\newtheorem{proposition}{Proposition}
\newtheorem{remark}{Remark}
\newcommand{\Xtar}{X_\mathrm{ter}}

\def\BibTeX{{\rm B\kern-.05em{\sc i\kern-.025em b}\kern-.08em
    T\kern-.1667em\lower.7ex\hbox{E}\kern-.125emX}}
\begin{document}
\title{Learning State-Action Control Barrier Functions for Model-Free Control under Constraints}
\author{Kanghui He, Shengling Shi, \IEEEmembership{Member, IEEE}, Ton van den Boom, and Bart De Schutter, \IEEEmembership{Fellow, IEEE}
	\thanks{This paper is part of a project that has received funding from the European Research Council (ERC) under the European Union’s Horizon 2020 research and innovation programme (Grant agreement No. 101018826 - CLariNet). }
	\thanks{Kanghui He is affiliated with the Department of Engineering Science, University of Oxford, OX1 3PJ Oxford, U.K., email: kanghui.he@eng.ox.ac.uk. Shengling Shi, Ton van den Boom, and Bart De Schutter are with Delft Center for Systems and Control, Delft University of Technology, Delft, The Netherlands, e-mail: $\{$shengling.shi, a.j.j.vandenBoom, b.deschutter$\}$@tudelft.nl.}}

\maketitle

\begin{abstract}
Ensuring constraint satisfaction for learning-based control is a critical challenge, especially in the model-free case. While safety filters address this challenge in the model-based setting, they typically rely on predictive models derived from physics or data. This reliance limits their applicability for advanced model-free learning control methods. To address this gap, we propose a new optimization-based control framework that determines safe control inputs directly from data. The benefit of the framework is that both the objective function and the safety constraint can be updated using model-free learning algorithms, enabling performance optimization while maintaining safety. As a key component, the concept of direct data-driven safety filters (3DSF) is first proposed. The framework employs a novel safety certificate, called the state-action control barrier function (SACBF). We present two approaches for synthesizing SACBFs, including learning from an expert controller and reinforcement learning. We develop a robustness framework that guarantees safety and recursive feasibility for the expert-guided approach under learning errors. The proposed control framework bridges the gap between model-free learning-based control and constrained control, by decoupling performance optimization from safety enforcement. Simulations on vehicle control illustrate the superior performance regarding constraint satisfaction and task achievement compared to model-based methods and reward shaping.
\end{abstract}


\section{Introduction}
\subsection{Background}
Learning-based control has achieved state-of-the-art performance in addressing complex problems, including applications in transportation systems \cite{haydari2020deep,he2024approximate} and robotics \cite{he2020composite,so2024train}. However, ensuring safety is still a challenging problem, particularly when an explicit model of the system is unavailable. Traditional model-based approaches to safety-critical control, such as model predictive control (MPC) \cite{borrelli2017predictive}, struggle with online computational efficiency and require an explicit prediction model, while emerging data-driven methods often lack well-understood safety guarantees. 

In control problems, safety is typically defined as maintaining state and input variables within given constraints \emph{throughout} the system’s evolution. A widely adopted approach to enforcing long-term safety is the use of safety filters, which provide a modular framework that can be applied to any control policy, even those without explicit safety considerations \cite{ames2016control,didier2024approximate,mestres2025control,he2024approximate,fisac2018general,choi2021robust,wabersich2021predictive}. The basic principle of safety filters is to post-process the input of a given control policy such that the resulting closed-loop system remains forward invariant w.r.t. the specified state and input constraints. The design of safety filters typically consists of two phases: an offline phase, where a safety certificate characterizing safe states is computed, and an online phase, where this certificate is incorporated as a constraint to modify potentially unsafe control actions from the reference controller. Learning-based approaches have increasingly been used to synthesize safety certificates, particularly for complex, nonlinear systems with non-convex state and input constraints. For a detailed overview, we refer the reader to relevant work in the literature \cite{robey2020learning,didier2024approximate,chen2024learning,westenbroek2021combining,taylor2020learning} as well as comprehensive surveys \cite{wabersich2023data,dawson2023safe}.

Despite recent progress, most safety filters rely on a mathematical model, which is either exactly derived from physical principles or approximately estimated. In most formulations, model information is required in both the offline and online phases.  Recently, there has been a growing number of approaches focusing on offline construction of safety certificates using only state transition data \cite{tan2023value}. However, these approaches still require an explicit (identified) model during online execution. This limitation is mainly due to the inherent property of existing safety certificates, whose safety conditions are defined solely over the state space.

To overcome the above limitation, data-driven safety filters have received significant attention. Almost all existing data-driven safety filters belong to \emph{indirect} data-driven methods \cite{cheng2019end,wabersich2023data,fisac2018general,westenbroek2021combining,taylor2020learning,dhiman2021control,wabersich2021predictive}. In indirect approaches, first system identification or disturbance estimation is performed and then a controller is learned based on the obtained model. In direct methods the controller is learned directly from data, bypassing the need for system identification. Indirect data-driven safety filters have one main issue: The errors arising from both model identification and certificate learning will compound, leading to a degradation in the safety performance of the filtered controller. 



\begin{figure}\label{compairson}
	\centering
	\includegraphics[width=240pt,clip]{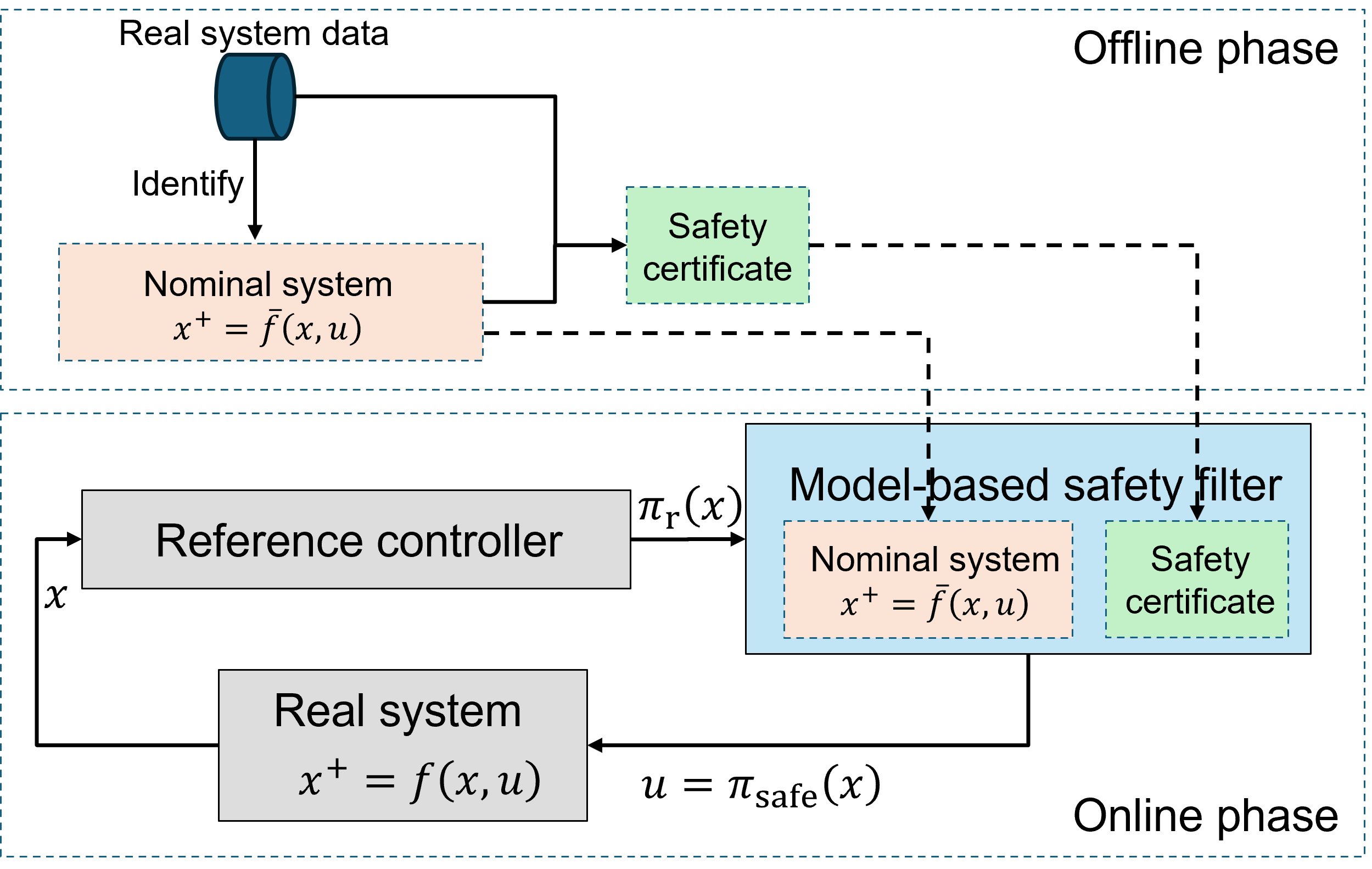}
	\hfil
	\includegraphics[width=240pt,clip]{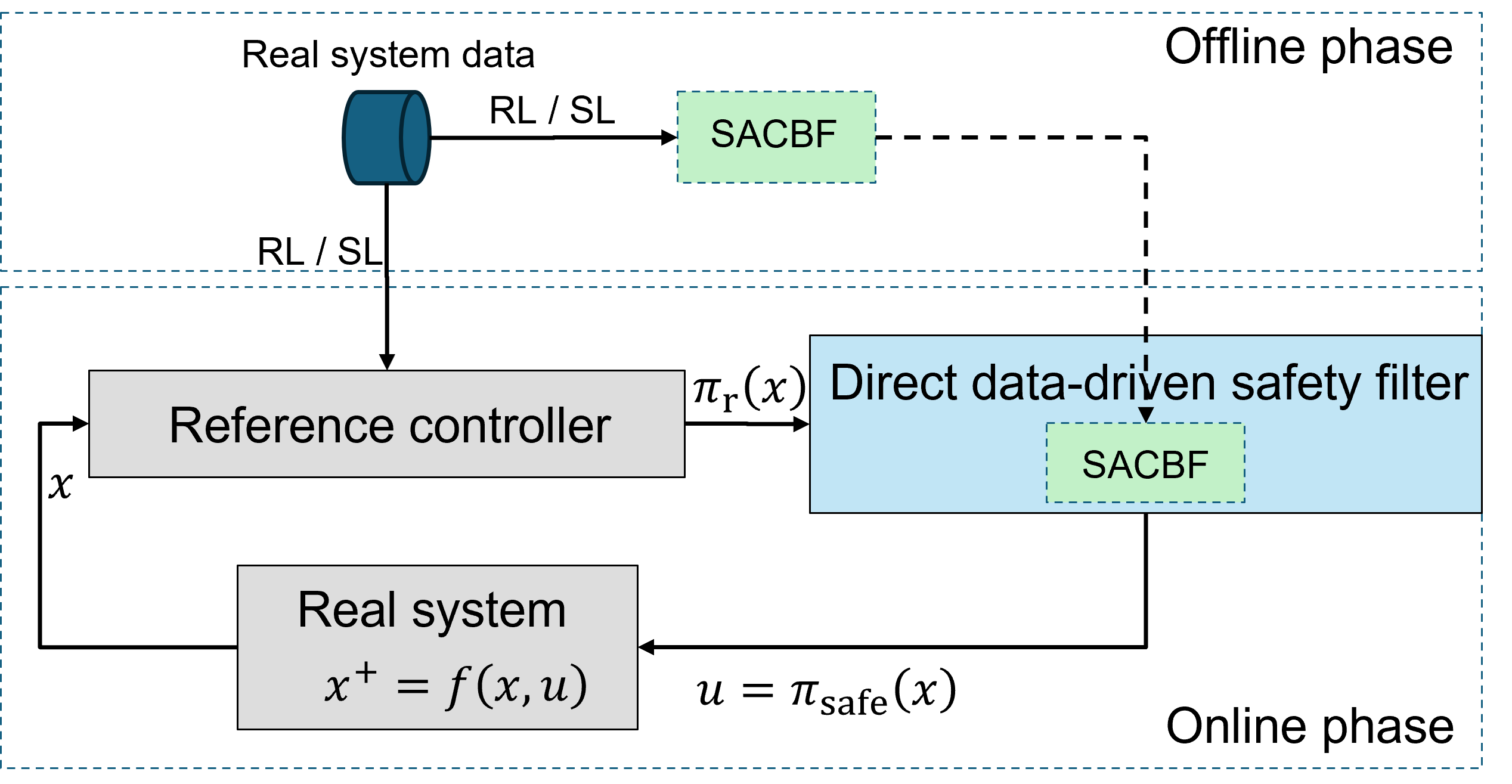}
	\caption{The comparison between the proposed 3DSF (top) and the existing indirect data-driven counterparts (bottom). }
\end{figure}

In this paper, we focus on designing a direct data-driven (model-free) safe control framework based on our previously proposed concept of \emph{state-action control barrier functions} (SACBFs) \cite{he2023state}. \hk{Unlike classical safety certificates, which only evaluate safety in the state space, SACBFs enable safety evaluation for each state-action pair, and the safety condition is enforced at the current state and action without propagating the dynamics forward. As a result, the corresponding safety filter can assess the safety of candidate control actions without explicitly querying or simulating the system dynamics.} The comparison between the proposed control framework and the existing indirect data-driven counterparts is illustrated in Fig. 1.

We next discuss several recently developed direct data-driven approaches for safe control. The authors of \cite{lavanakul2024safety} learn discriminating hyper-planes to directly regulate the control inputs. However, this method is limited to linear safety constraints on inputs and lacks formal guarantees regarding constraint satisfaction. \hk{\cite{lanza2025model} develops CBF-based control without a model by using techniques from funnel control. The considered safety constraint is specifically introduced to keep the tracking error within a prescribed bound. In comparison, our method directly handles general nonlinear state constraints and can optimize a given performance criterion. Similarly, \cite{molnar2021model} achieves CBF-based safe control for robotic systems without knowing the dynamics. Its key idea is to exploit the known kinematics to generate a safe velocity and then employ a model-free controller to track this velocity. Consequently, the approach in \cite{molnar2021model} relies on the particular structure of robotic systems. In contrast, our proposed method does not rely on such assumptions and applies to general nonlinear systems.}

\subsection{Contributions}
The paper contributes to the state of the art in the following aspects:

\noindent(1) \hk{We propose a safe control framework based on SACBFs that combines the advantages of
learning-based control (applicability to model-free cases) and
optimization-based control (explicit inclusion of safety constraints).
The core mechanism relies on our previously introduced
safety certificate called the SACBF \cite{he2023state}. We extend existing learning CBF approaches, including learning from an expert controller \cite{robey2020learning} and RL \cite{fisac2019bridging}, to learning SACBFs. The framework can be trained and implemented using
only state transition data, without system identification, and
can integrate with various learning-based controller synthesis
methods.}

\noindent(2) \hk{We develop an expert-guided approach for learning SACBFs directly from transition data. We establish a converse result guaranteeing existence under the expert-policy assumption and develop an error-to-state safety analysis to characterize recursive feasibility and constraint satisfaction under learning errors. As an extension, we further employ RL to synthesize SACBFs when an expert controller is unavailable.}

\noindent(3) \hk{Based on the proposed control framework, we present a modular strategy for extending classical unconstrained value-based RL algorithms to constrained ones. This modular construction separates controller synthesis into two components: learning a performance value function and learning an SACBF that encodes safety.}

\section{Preliminaries and problem formulation} \label{pre}

\subsection{Problem formulation}
We consider a deterministic discrete-time nonlinear system
\begin{equation}\label{system}
	x_{t+1}=f\left(x_t, u_t\right), \quad t=0,1, \ldots,
\end{equation}
where $x_t \in \mathcal{X} \subseteq \mathbb{R}^{n_x}$ and $u_t \in \mathcal{U}\subseteq \mathbb{R}^{n_u}$ are the state and the input at time step $t$, and $f(\cdot,\cdot):\mathcal{X} \times \mathcal{U} \to \mathcal{X}$ is a continuous function. We consider a constrained optimal control problem, in which the states and inputs should satisfy time-invariant constraints: $x_t \in X := \{x \in \mathcal{X}| h(x) \leq 0\}$ and  $u_t \in U \subseteq \mathcal{U}$ for all time steps. Here, $h(\cdot): \mathcal{X} \to \mathbb{R}$ is a scalar-valued continuous function that defines the state constraint. The sets $U$ and $\mathcal{X}$ are compact sets. We do not assume the knowledge of the explicit form of the system dynamics $f(\cdot,\cdot)$. Instead, we require the availability of transition samples $(x_t, u_t, x_{t+1})$, achieved through simulation or experimental methods. Various sampling strategies, including random sampling and grid-based sampling, may be employed to obtain the transition triples.

For any initial state $x \in \mathcal{X}$ and any deterministic control policy $\pi: \mathcal{X} \to \mathcal{U}$, we define an infinite-horizon cost $J^\pi\left(\cdot \right): \mathcal{X} \to [0,\infty)$:
\begin{equation}\label{value function}
	J^\pi\left(x\right):= \lim _{k \rightarrow \infty}  \sum_{t=0}^k \gamma^t g\left(x_t, \pi\left(x_t\right)\right)\; \mathrm{s.t.} \; \eqref{system} \;\mathrm{and}\;x_0 =x,
\end{equation}
where the stage cost $g$ is non-negative and upper-bounded and $\gamma \in (0,1)$ is a discount factor. The objective is to find an optimal deterministic policy $\pi^*(\cdot) : \mathcal{X} \to \mathcal{U}$ that solves the following infinite-horizon optimal control problem:
\begin{align}\label{infinite}
	J^*\left(x\right) := \inf _{\pi}  J^\pi\left(x\right)\;\mathrm{s.t.}\; h(x_t) \leq 0, \pi(x_t) \in U, t=0, 1, \ldots,
\end{align}
where $J^*(\cdot): \mathcal{X} \to [0,\infty)$ is the optimal value function.

\hk{The exact form of the optimal policy $\pi^*$ is difficult to compute because the number of constraints in \eqref{infinite} is infinite, and the closed-form expression of the objective function and the model are unknown. In this paper, the following problem will therefore be addressed.}

\hk{\textbf{Problem:} How to design a data-driven controller that satisfies the state/input constraints while optimizing the prescribed performance objective \eqref{value function} without explicit knowledge of the system model $f$?}

\hk{In the remainder, we propose an optimization-based control method based on the previously developed SACBF. The SACBF is used to enforce the state constraints, while the objective function of the optimization problem is used to improve control performance. In our previous work \cite{he2023state}, we proposed a model-based approach for constructing a valid SACBF by supervised learning. In this paper, we consider two model-free approaches for synthesizing SACBFs: an expert-guided learning approach and an RL-based approach. In the absence of approximation errors, constraint satisfaction of the resulting controller is guaranteed. In the presence of learning errors, we further develop a constraint-tightening and SACBF-constraint-relaxation scheme for the expert-guided approach. When the approximation error is upper-bounded, this scheme ensures recursive feasibility and satisfaction of the original constraints.}

\subsection{Control barrier functions for safety}
To ensure that the learned policy satisfy the constraints in \eqref{infinite}, two different classes of methods, including cost shaping \cite{massiani2022safe} and using safety certificates \cite{agrawal2017discrete,didier2024approximate,ames2016control}, have been developed in recent papers. In contrast to the cost shaping method, which usually does not provide strict safety guarantees, using safety certificates can impose a constraint on the control input to ensure safety. The CBF is one of the most popular safety certificates. 

\begin{definition}[Discrete-time control barrier function \cite{agrawal2017discrete}]\label{CBF definition}
	A function $B(\cdot):\mathcal{X} \to \mathbb{R}$ is called a \emph{control barrier function} (CBF) with a corresponding safe set $\mathcal{S}_B:=\left\{x \in \mathbb{R}^{n_x}| B(x) \leq 0\right\} \subseteq \mathcal{X}$, if $\mathcal{S}_B$ is non-empty, $h(x)\leq B(x),\;\forall x \in \mathcal{X}$, and if there exists a $\beta_B \in (0,1]$ such that
	\begin{equation}\label{CBFex}
		\min _{u \in U} B(f(x, u)) \leq \beta_B B(x),\;\forall x \in \mathcal{X}.
	\end{equation}
\end{definition}

\hk{For a conventional state-based CBF $B$, safety of a candidate input is determined through condition \eqref{CBFex} involving the successor state. Therefore, evaluating whether a candidate input is safe requires the dynamics $f$, or an identified/estimated model of $f$, at runtime.} This limitation significantly restricts the application of existing safe filters and CBFs in RL algorithms that do not require an explicit model. Moreover, the mismatch between the identified model and the real system, along with inaccuracies in learning CBFs, jointly affect the safety performance of $\pi$. 

\subsection{State-action control barrier function}
We consider the definition of state-action CBFs, which we have introduced in \cite{he2023state}.
\begin{definition}[State-action control barrier function (SACBF)]\label{SACBF definition}
	A function $Q^B(\cdot,\cdot) : \mathcal{X} \times U \to \mathbb{R}$ is called a \emph{state-action control barrier function} (SACBF) with a corresponding safe set $\mathcal{S}_Q$ of states, if the pair $(Q^B, \; \mathcal{S}_Q)$ satisfies the following conditions:
	
	\noindent(i) $\mathcal{S}_Q$ is non-empty, and $h(x) \leq 0,\; \forall x \in \mathcal{S}_Q$.
	
		\noindent(ii) $\min_{u\in U} Q^B(x,u) \leq 0,\;\forall x\in \mathcal{S}_Q$.
	
		\noindent(iii) For any $x \in \mathcal{S}_Q$, any $u \in U$ satisfying $Q^B(x,u) \leq 0$ ensures that $f(x,u) \in \mathcal{S}_Q$.
\end{definition}

\hk{In contrast, an SACBF $Q^B$ directly encodes the safety admissibility of each state-action pair $(x,u)$ by $Q^B(x,u) \leq 0$. Under this condition, the applied input ensures safety according to Lemma \ref{lemma2} below. Therefore, the main advantage of SACBFs is that candidate control inputs can be evaluated and filtered without propagating them through dynamics.}

\section{Control parameterization} 
We propose a new optimization-based control framework that does not contain $f$:
\begin{align}\label{DDfilter}
	\pi(x): =\arg\min_{u\in U}&\; Q(x,u) \nonumber\\
	\mathrm{s.t.}\;&  Q^B(x,u)\leq 0 ,
\end{align}
where $Q^B(\cdot,\cdot): \mathcal{X}\times \mathcal{U} \to \mathbb{R}$ is an SACBF and $Q:\mathcal{X}\times \mathcal{U} \to \mathbb{R}$ is a problem-dependent objective function.

We have the following properties for \eqref{DDfilter}:
\begin{lemma}[Safety of SACBFs (adapted from Lemma 1 of \cite{he2023state})]\label{lemma2}
	Let $Q^B$ be an SACBF with the safe set $\mathcal{S}_Q$. The closed-loop system \eqref{system} under $\pi$ renders $\mathcal{S}_Q$ positively invariant. As a result, the trajectories of \eqref{system}, starting from $x_0 \in \mathcal{S}_Q$, controlled by $\pi$, satisfy $h(x_t) \leq 0$ and $u_t \in U,\;\forall t\in \mathbb{N}$.
\end{lemma}

Analogous to CBFs, SACBFs are difficult to compute exactly for nonlinear problems. In this paper, we use learning-based methods to approximate them. A tractable way is to parameterize the functions $Q$ and $Q^B$ by $Q_\theta$ and $Q^B_\omega$ with the parameters $\theta$ and $\omega$, and then update these parameters. Instead of simultaneously updating $Q$ and $Q^B$, we prefer to first learn an SACBF $Q^B_\omega$ (Sections \ref{section expert} and \ref{section RL SACBF}) and subsequently integrate it into the learning of the objective function $Q_\theta$ (Section \ref{section_refine}), because the safety certificate can be learned independently of the task-specific performance objective.

The policy $\pi$ in \eqref{DDfilter} with the parameterization $Q_\theta$ and $Q^B_\omega$ is denoted by $\pi_{\theta, \omega}$.

\section{Learning state-action control barrier functions}\label{section learning SACBF}

In this section, we present three approaches that differ in their assumptions and the strength of the resulting safety guarantees. We begin with the most general method (RL approach), which does not require any assumptions. This method relies on off-the-shelf value-based RL to approximate $Q^B$. However, the learning error may in principle become unbounded. Consequently, the RL approach cannot offer formal safety guarantees. Motivated by this limitation, we then introduce two additional approaches that impose stronger assumptions. The second approach, inspired by \cite{robey2020learning}, employs sampling-based methods to approximate the solution of a semi-infinite optimization problem. Such a solution is guaranteed to lead to a valid SACBF. This approach requires prior knowledge of a safety control policy. For the last one, we assume the knowledge of a valid CBF and use SL to obtain an SACBF. These assumptions restrict generality but allow the approximation errors to be explicitly characterized, enabling us to derive provable and rigorous safety guarantees.

\subsection{Learning SACBFs via RL}\label{section RL SACBF}

We observe that the proposed SACBFs show some similarities with state-action value functions (Q functions) in RL. These similarities motivate us to use RL methods to synthesize SACBFs. We consider the optimal value function of the following optimal control problem:
\begin{align}\label{CBF generator}
	B^*(x) := &\min _{ \{(x_t,u_t)\}_{t=0}^{\infty}} \max _{t \in \mathbb{N}}   h(x_t )  \nonumber\\
	&\quad \mathrm{s.t.}\; (1),\;\;u_t \in U,\;t \in \mathbb{N}, \text{ and }x_0 = x,
\end{align}
which is a typical reachability problem \cite{fisac2019bridging} in the discrete-time domain. The optimal value function $B^*$, which can be obtained by dynamic programming, is a CBF with the safe set being the maximal control-invariant set \cite{choi2021robust}. 

\begin{proposition}\label{CBF relation}
	Under the constraints $h(x)\leq 0$ and $u\in U$, suppose that $B$ is a CBF satisfying \eqref{CBFex} with $\beta_B \in [0,1]$, and that the safe set is $\mathcal{S}_B$. Then, $Q^B$ defined by
	\begin{equation}\label{relation}
		Q^B(x,u):= \max\{h(x),\;\frac{1}{\beta_B} B(f(x,u))\}
	\end{equation}
	is an SACBF with the safe set $\mathcal{S}_Q = \mathcal{S}_B$.
\end{proposition}
\begin{proof}
	From the relation between $Q^B$ and $B$, we can get that $h(x) \leq  \min_{u\in U} Q^B(x,u) \leq 0,\; \forall x\in \mathcal{S}_B$. The above arguments prove that $Q^B$ obeys (i) and (ii) of Definition \ref{SACBF definition}. Furthermore, for any $x \in \mathcal{S}_B$, if $Q^B(x,u) \leq 0$, we have $B(f(x,u))\leq 0$ and thus $f(x,u) \in \mathcal{S}_B$. This finishes the proof. 
\end{proof}
    
Using Proposition \ref{CBF relation} and letting $\beta_B =1$, we know that the state-action optimal value function of \eqref{CBF generator}, defined by
	\begin{equation}\label{Q safe function}
		Q^{B*}(x,u) := \max\{h(x),\; B^*(f(x,u))\},
	\end{equation}
	is an SACBF, of which the safe set is also the maximal control-invariant set.



%
Combining \eqref{CBF generator} and \eqref{Q safe function}, it is observed that $Q^{B*}$ is one solution of the Bellman optimality equation \cite{busoniu2017reinforcement}:
\begin{equation}\label{bellman}
	Q^{B}(x,u) = \max\{h(x), \min_{u^+\in U} Q^{B}(f(x,u),u^+) \}.
\end{equation}

For general nonlinear systems, accurately computing $Q^{B*}$ from \eqref{CBF generator} and \eqref{Q safe function} is nearly impossible. However, by employing \eqref{bellman} and the theoretical results in Proposition \ref{CBF relation}, we can adapt off-the-shelf value-based RL methods such as temporal difference learning methods \cite{bansal2021deepreach,watkins1992q} and neural fitted Q iteration (FQI) \cite{riedmiller2005neural} to approximate $Q^{B*}$. 

Note that \eqref{bellman} may have multiple solutions and that most RL methods based on Bellman iteration could approach any solution of \eqref{bellman}. The following result eliminates our concern by stating that any solution of \eqref{bellman}, satisfying $\min_{x \in \mathcal{X}, u\in U} Q^{B}(x,u) \leq 0$, is an SACBF. 

\begin{proposition}\label{SACBF Bellman}
	Consider the Bellman optimality equation \eqref{bellman}. Any solution $Q^{B}$ satisfying $\min_{x \in \mathcal{X}, u\in U} Q^{B}(x,u) \leq 0$ is an SACBF with the safe set $\mathcal{S}_Q = \{x \in\mathbb{R}^{n_x} | \min_{u\in U} Q^{B}(x,u) \leq 0\}$. Moreover, $Q^{B*}$ is the smallest solution, i,e., $Q^{B*}(x,u) \leq Q^{B}(x,u),\;\forall(x,u)\in \mathcal{X} \times U$ holds for any solution $Q^{B}$. 
\end{proposition}
\begin{proof}
It follows from $\min_{x \in \mathcal{X}, u\in U} Q^{B}(x,u) \leq 0$ that $\mathcal{S}_Q$ is non-empty. From \eqref{bellman}, we deduce that $h(x) \leq  \min_{u\in U} Q^{B}(x,u) \leq 0,\; \forall x\in \mathcal{S}_Q$. The above arguments prove that $Q^{B}$ satisfies (i) of Definition \ref{SACBF definition}. The second condition of Definition \ref{SACBF definition} directly follows from $\mathcal{S}_Q = \{x \in\mathbb{R}^{n_x} | \min_{u\in U} Q^{B}(x,u) \leq 0\}$. Finally, if $Q^B(x,u)\leq 0$, by \eqref{bellman} we have $\min_{u^+\in U} Q^{B}(f(x,u),u^+) \leq 0$, which means that the successor state $f(x,u) \in \mathcal{S}_Q$. 

To prove that  $Q^{B*}$ is the smallest solution, consider any solution $Q^{B}$. It follows from \eqref{bellman} that 
\begin{equation} \label{proof1}
	h(x) \leq Q^{B}(x,u),\;\forall(x,u)\in \mathcal{X} \times U.
\end{equation}
By applying the Bellman operation $\Gamma_B(\cdot) := \max\{h(x), \min_{u^+\in U} (\cdot)(f(x,u),u^+)\}$ to both sides of \eqref{proof1} infinitely many times and by $\Gamma_B Q^{B} =  Q^{B}$, $\Gamma^\infty_B h =  Q^{B*}$, and the monotonicity of  $\Gamma_B$ \cite{bertsekas2022abstract}, we obtain that $Q^{B*}(x,u) \leq Q^{B}(x,u),\;\forall(x,u)\in \mathcal{X} \times U$. This completes the proof of Proposition \ref{SACBF Bellman}.
\end{proof}

Following the temporal difference learning method in \cite{bansal2021deepreach}, the training loss for the candidate SACBF $Q^B_{\omega}$ is designed as 
\begin{align}\label{TDerror}
	l(\mathcal{D}, \omega) &= l_1(\mathcal{D}, \omega)+ \rho  l_2(\mathcal{D}, \omega), \;\mathrm{with}\nonumber\\
	l_1(\mathcal{D}, \omega) &= \sum_{(x, u, x^+) \in \mathcal{D}}\left.(\max\{h(x), \min_{u^+\in U} Q^B_{\omega}(x^+,u^+) \}\right. \nonumber\\
		&\quad \quad \quad \quad \quad \quad  \left. - Q^B_{\omega}(x,u)\right)^2\nonumber\\
	l_2(\mathcal{D}, \omega) & = \sum_{(x, u,x^+) \in \mathcal{D}}  Q^B_{\omega}(x,u),
\end{align}
where $\mathcal{D} = \{(x^{(i)}, u^{(i)}, f(x^{(i)},u^{(i)}))\}^N_{i=1}$ represents the collection of state transition triples. Intuitively, minimizing the loss $l_1$ encourages $Q^B_{\omega}$ to reduce the temporal difference, thereby approaching the solution of \eqref{bellman}. Meanwhile, as the Bellman equation \eqref{SACBF Bellman} may have multiple solutions, the loss $l_2$ is introduced to guide $Q^B_{\omega}$ toward the smallest solution of \eqref{bellman}, which has the largest safe set. The positive parameter $\rho$ balances the relative importance of the two loss functions.

\begin{remark}
    \hk{For the initialization $Q_0 (x,u) := \max\{h(x),\;h(f(x,u))\}$, the exact value iteration without approximation errors generates a nondecreasing sequence converging pointwise to $Q^{B*}$, a valid SACBF with the largest safe set \cite{bertsekas2022abstract}. However, since the reachability problem is undiscounted, most RL algorithms using function approximations cannot, in general, guarantee convergence. Tractable approaches to alleviate this issue include introducing a discount factor \cite{fisac2019bridging,so2024train} and considering the reachability to a known invariant set \cite{choi2023forward}. In the case study, we compare the undiscounted and discounted formulations, both of which exhibit numerical stability when learning the SACBF.}
\end{remark}
\begin{remark}
    \hk{Value-based RL and temporal-difference methods have already been used in the literature to learn standard CBFs or reachability value functions. The RL-based synthesis method in this section is an extension of existing RL techniques for learning standard CBFs/reachability functions \cite{fisac2019bridging,bansal2021deepreach} to the learning of SACBFs, which in turn enables the proposed direct data-driven safe control.}
\end{remark}

\subsection{Learning SACBFs from an expert controller}  \label{section expert}

The RL approach usually does not have performance guarantees in theory. Towards the goal of learning a valid SACBF, like in \cite{robey2020learning} where CBFs are synthesized from expert demonstrations, we assume the availability of an expert safe controller $\pi_\mathrm{s}(\cdot): \mathcal{X} \to \mathcal{U}$ in this subsection. Formally, we have the following assumption.
\begin{assumption}\label{A2}
	There exists a continuous policy $\pi_\mathrm{s}$ and a compact set $\mathcal{S}_0 \subseteq X$ ($\mathcal{S}_0$ can be unknown) such that with the initial condition $x_0 \in \mathcal{S}_0$, the state-input trajectories of the system \eqref{system} with $\pi_\mathrm{s}$ always stay in $\mathcal{S}_0 \times U$ and such that the states reach a non-empty terminal set $\Xtar\subseteq X$ within at most $T\in \mathbb{N}$ steps.
\end{assumption}
Note that we do not require the explicit form of the safe policy. This assumption is reasonable in certain scenarios, such as emergency braking in driving, where safety is achieved by feedback on state variables (e.g., position and velocity) without an explicit model. Other types of expert safe controllers include artificial potential fields (APF) \cite{warren1989global} for obstacle avoidance and safe controllers derived from Hamilton-Jacobi reachability, which have recently been extended to model-free settings \cite{fisac2019bridging}. \hk{When discussing safety, it is standard to define it relative to a specified set of initial conditions \cite{robey2020learning}. Accordingly, Assumption \ref{A2} only requires the existence of a nonempty set $S_0$. The expert controller is not required to be safe outside $S_0$, and $S_0$ itself does not need to be known a priori.}  \hk{The terminal set is introduced to certify whether trajectories are safe over the entire system evolution. Specifically, the sampled trajectory is required to satisfy the constraints over the considered finite horizon and to reach the terminal set before the terminal time. The terminal set is chosen such that, once it is reached, either the control task is completed or safety can be maintained thereafter. In the latter case, the terminal set is control invariant \cite{borrelli2017predictive}, meaning that there exists a controller that keeps the subsequent trajectory within the set. Construction of such a control-invariant set does not necessarily require the system model and can be achieved by existing data-driven methods \cite{chakrabarty2018data}.}

Under Assumption \ref{A2}, we formulate the synthesis of $Q^B$ as the following optimization problem:

\begin{subequations}\label{robust optimization}
	\begin{align}
		\min_{Q^B,q}&\; \int_{x \in \mathcal{X}} q(x) dx  \label{robust optimization a}\\
		\text{s.t.}\;&0 \leq q(x),\;\forall x \in \mathcal{X} \setminus \mathcal{S}_0 \label{robust optimization b}\\
		& Q^B(x, \pi_\mathrm{s}(x)) \leq \beta q(x), \forall x \in \mathcal{S}_0 \label{robust optimization c} \\
		& q(f(x,u)) \leq Q^B(x, u), \forall (x,u) \in \mathcal{S}_0 \times U \label{robust optimization d} 
	\end{align}
\end{subequations}
where $Q^B$ and $q(\cdot) : \mathcal{X} \to \mathbb{R}$ are continuous in their domain, $\beta \in [0,1)$ is a tuning parameter, and the objective of minimizing $q$ contributes to maximizing the volume of the safe set.

One relevant question is whether the semi-infinite optimization problem \eqref{robust optimization} is actually feasible. In the following proposition, we address this concern by presenting a converse SACBF result. This result shows that if a safe policy exists (Assumption \ref{A2}), then an SACBF must also exist. Consequently, the feasibility of problem \eqref{robust optimization} is guaranteed. Moreover, the proposition establishes that any optimal solution to \eqref{robust optimization} yields a valid SACBF with a non-empty safe set.
\begin{proposition}[Converse SACBFs]\label{expert}
	Under Assumption \ref{A2}, there exists a $\beta \in  [0,1)$ such that the problem \eqref{robust optimization} is feasible and such that any optimal solution pair $(Q^{B*},q^*)$ ensures that $Q^{B*}$ is an SACBF with a non-empty safe set $\mathcal{S}^*_{Q} =  \{x \in\mathcal{X} | q^*(x) \leq 0\}$.
\end{proposition}

\begin{proof}
	The proof consists of two parts. In the first part, we will prove the feasibility of \eqref{robust optimization}. By fixing $Q(x,u) = q(f(x,u))$, which satisfies \eqref{robust optimization d}, the constraint \eqref{robust optimization c} becomes $q(f(x,\pi_\mathrm{s}(x))) \leq \beta q(x),\;\forall x \in \mathcal{S}_0$. 
	
	Since $\mathcal{S}_0$ is compact, there exists a continuous function $q$ such that \eqref{robust optimization b} and
	\begin{align}\label{level set}
		\mathcal{S}_0 & =\left\{x \in \mathcal{X}| q(x) \leq 0\right\}, \nonumber\\
		\partial \mathcal{S}_0 & =\left\{x \in \mathcal{X}| q(x)=0\right\}, \nonumber\\
		\operatorname{Int}(\mathcal{S}_0) & =\left\{x \in \mathcal{X}| q(x)<0\right\}
	\end{align}
hold. Actually, the above statement can be proved by considering $q(x) = \max\{h(x),\;\bar q(x)\}$, where $\bar{q}$ is a distance function defined by
	\begin{equation}\label{distance}
		\bar{q}(x)  = \begin{cases}-\inf _{y \in \partial \mathcal{S}_0}\|x-y\|\text{ if } x \in \mathcal{S}_0 \\ \inf _{y \in \partial \mathcal{S}_0}\|x-y\|\text{ otherwise}.\end{cases}
	\end{equation}
	
	With the above definitions, we follow the proof of the Converse CBF Theorem in the continuous time domain \cite[Proposition 3]{ames2016control} to prove the feasibility of \eqref{robust optimization}. In particular, for any continuous $q$ such that \eqref{robust optimization b} and \eqref{level set} hold, define a function $\alpha(\cdot):[0,\infty) \to \mathbb{R}$ by
	\begin{equation}\label{alpha}
		\alpha(r) = \sup_{\{x\in\mathcal{X}| -r \leq q(x) \leq 0\}} q(f(x,\pi_\mathrm{s}(x))) - q(x).
	\end{equation}
	Since $\{x\in\mathcal{X}| -r \leq q(x) \leq 0\}$ is compact and $q$, $f$, and $\pi_\mathrm{s}$ are continuous, $\alpha$ is well-defined and non-decreasing, and satisfies $q(f(x,\pi_\mathrm{s}(x))) - q(x) \leq \alpha (-q(x)),\;\forall x \in \mathcal{S}_0$.
	
	When $x$ is such that $q(x)=0$, i.e., $x \in \partial \mathcal{S}_0$, by the invariance of $\mathcal{S}_0$, we have $q(f(x,\pi_\mathrm{s}(x))) \leq 0$, and consequently $\alpha (0) \leq 0$. Meanwhile, by taking the extreme values $0$ and $-r$ of $q(f(x,\pi_\mathrm{s}(x)))$ and $q(x)$ respectively, we have $\min_{r>0}\frac{\alpha(r)}{r} \leq \min_{r>0}\frac{0-(-r)}{r} =1$, which, combined with $\alpha (0) \leq 0$, implies that there exists a class $\mathcal{K}$ function $\bar{\alpha}$ upper-bounding $\alpha$ and satisfying 
	\begin{align}
		&q(f(x,\pi_\mathrm{s}(x)))-q(x) \leq \bar\alpha(-q(x)), \forall x \in \mathcal{S}_0 ,\nonumber\\
		&\min_{r>0}\frac{\bar \alpha(r)}{r} \leq 1.
	\end{align}
Letting
\begin{equation}
	\beta = 1 - \min_{r>0}\frac{\bar \alpha(r)}{r} \in [0,1)
\end{equation}
proves the feasibility of \eqref{robust optimization}.
	
	In the second part of the proof, we will show that the optimal solution to \eqref{robust optimization} yields an SACBF. To this end, we first establish the non-emptiness of $\mathcal{S}^*_{Q}$ by contradiction. Suppose that $\mathcal{S}^*_{Q}$ is empty, i.e., $q^*(x) > 0,\;\forall x \in \mathcal{X}$. Consider a new function $q'(\cdot):\mathcal{X} \to \mathbb{R}$ defined by $q'(x) =  \min\left\{q^*(x),\;\max\{h(x),\;\bar{q}(x)\}\right\}$, where $\bar{q}$ is defined by \eqref{distance}. Based on the proof of the feasibility of \eqref{robust optimization}, $\max\{h(x),\;\bar{q}(x)\}$ satisfies all the constraints in \eqref{robust optimization}. Therefore, we have
	\begin{equation*}
		q'(x^+) \!\leq\! \min\left\{ \beta q^*(x), \beta\max\{h(x),\bar{q}(x)\}\right\} \!\leq \!\beta q'(x),\;\forall x \in \mathcal{S}_0
	\end{equation*}
	with $x^+ =  f(x,\pi_\mathrm{s}(x))$. This implies that the new function $q'$ satisfies all the constraints in \eqref{robust optimization}. Moreover, since the zero sub-level set of $\max\{h,\;\bar{q}\}$ is non-empty, the zero sub-level set of $q'$ is non-empty as well. Furthermore, it strictly holds that $\int_{x\in \mathcal{X}} q'(x) dx < \int_{x\in \mathcal{X}} q^*(x) dx$ because (i) $q'(x) \leq q^*(x),\;\forall x \in \mathcal{X}$, and (ii) $q^*>0$ and $q'\leq 0$ in the zero sub-level set of $q'$. These arguments prove that the new function $q'$ provides a better solution than $q^*$, which contradicts the optimality of $q^*$.
	
	The condition (i) of Definition \ref{SACBF definition} holds because of \eqref{robust optimization b} and the non-emptiness of $\mathcal{S}^*_{Q}$. Moreover, $\forall x \in \mathcal{S}^*_{Q}$, it follows from \eqref{robust optimization c} that $\min_{u\in U} Q^{B*} (x,u) \leq Q^{B*} (x,\pi_\mathrm{s}(x)) \leq \beta q(x) \leq 0$, meaning that the condition (ii) holds. Furthermore, \eqref{robust optimization d} enforces the condition (iii). These complete the proof that $Q^{B*}$ is an SACBF.
\end{proof}

For computational tractability, $Q^B$ and $q$ in \eqref{robust optimization} are parameterized as $Q^B_\omega$ and $q_\omega$, with all the parameters condensed in $\omega$. As a result, the search space becomes the parameter space of $\omega$. To solve \eqref{robust optimization}, we use state transition samples to relax the constraints that should hold in continuous spaces to constraints that only need to hold for the samples. 

In particular, let $\{x^{(i)}_0\}^{N}_{i=1}$ denote the set of initial states. Here, $N$ is the number of samples. For each $x^{(i)}_0$, we apply $\pi_\mathrm{s}$ to \eqref{system} for $T$ time steps and get the state-input trajectory $\left\{(x^{(i)}_t,\pi_\mathrm{s}(x^{(i)}_t))\right\}^{T}_{t=0}$. If this trajectory obeys $x^{(i)}_t \in X,\;\pi_\mathrm{s}(x^{(i)}_t) \in U,\;\forall t\in \mathbb{N}_T$, and $x^{(i)}_T \in \Xtar$, then $x^{(i)}_t \in \mathcal{S}_0,\; \forall t\in \mathbb{N}_T$, where $\mathbb{N}_T:= \{0,1,...,T\}$. Otherwise, $x^{(i)}_t \notin \mathcal{S}_0,\; \forall t\in \mathbb{N}_T$. The above checking procedure allows us to separate the state trajectories according to whether or not they are in $\mathcal{S}_0$, without knowing $\mathcal{S}_0$. Then, we define the set $\mathcal{S}_\mathrm{s}: = \{ x^{(i)}_t | x^{(i)}_t \in X,\;\pi_\mathrm{s}(x^{(i)}_t) \in U,\;\forall t\in \mathbb{N}_T$, $x^{(i)}_T \in \Xtar,\; \forall i \in \{1,2,...,N\}\}$, which contains the safe state trajectories. Meanwhile, we define $\mathcal{X}_\mathrm{s}:=\{x^{(i)}_t\}^{N\;\;\; T}_{i=1\;t=0}$.

For the optimization problem \eqref{robust optimization}, we replace the infinite sets $\mathcal{S}_0$, $\mathcal{X}$, and $\mathcal{S}_0 \times U$ by the finite sets $\mathcal{S}_\mathrm{s}$, $\mathcal{X}_\mathrm{s}$, and $\mathcal{S}_\mathrm{s}\times U_\mathrm{s}$, respectively. Here $U_\mathrm{s}$ is the set of inputs sampled in $U$ through some sampling strategy\footnote{We do not assume i.i.d. samples.}. As a result, \eqref{robust optimization} is simplified to a constrained optimization problem with a finite number of constraints. If $Q^B_{\omega}$ and $q$ are linear in $\omega$, the problem \eqref{robust optimization} becomes linear and can be solved efficiently. In a more general case such as when these functions are represented by deep neural networks, a tractable solution involves transforming the constraints into soft penalties and adding them into the objective.

\begin{remark}
    \hk{By learning an SACBF from the expert controller, the resulting optimization-based controller can therefore optimize an independent performance objective subject to the learned SACBF constraint, potentially reducing the conservatism of directly using the expert policy. A quantitative analysis of the suboptimality regarding the infinite-horizon cost \eqref{value function} is an important topic for future work.}
\end{remark}

\begin{remark}
	The above sampling-based approach does not guarantee the validity of the approximated solution. \cite{robey2020learning}, which studies learning CBFs, suggests leveraging the Lipschitz properties of $f$ and the approximated solution for principled constraint tightening. However, estimating the Lipschitz constants when $Q^B$ is parametrized by a neural network is non-trivial, and obtaining non-conservative Lipschitz constants for the unknown function $f$ is questionable. To address this limitation, we will propose a new tightening approach that exploits the inherent robustness of the SACBF. This will be elaborated on in Section \ref{section performance}.
\end{remark}

\subsection{Supervised learning SACBFs from CBFs} \label{section supervised}
At the end of this section, we present the most sample-efficient learning-based method for obtaining an SACBF, based on a rather restrictive assumption that we have a CBF at hand. 

\begin{assumption}\label{A1}
	A CBF $B(\cdot):\mathcal{X} \to \mathbb{R}$ is known.
\end{assumption}

In certain applications, such as adaptive cruise control \cite{ames2016control}, CBFs can be manually designed based on state constraints, such as the distance between adjacent vehicles. Moreover, there are mature techniques for synthesizing valid CBFs without model information \cite{tan2023value,so2024train,zhang2025discrete}. However, as noted in the introduction, such CBFs cannot be directly used to design a safety filter when the model is not fully known.

Proposition \ref{CBF relation} implies a straightforward way to learn an SACBF. Initially, a collection of transition triples $\{(x^{(i)}, u^{(i)}, f(x^{(i)},u^{(i)}))\}^N_{i=1}$, is generated from the state space $\mathcal{X}$ and the control space $U$. Subsequently, the labels $Q^B(x^{(i)},u^{(i)})$ can be computed utilizing \eqref{relation}. Following this, a regression model $Q^B_\omega$ is constructed and trained to minimize the empirical loss, specifically the mean squared error between $Q^B_\omega$ and the computed labels. 

\begin{remark}
    The SL method is intended for scenarios where a valid CBF is already available, but the model is unknown. In such a case, even when a CBF is known, traditional safety filter design based on CBFs still depends on a nominal model (typically obtained through system identification) to enforce the invariance condition \cite{wabersich2023data}. In contrast, the proposed SL method offers an alternative by learning a scalar SACBF to enforce invariance directly, avoiding the need for system identification.
\end{remark}

\section{Performance analysis and guarantee under learning errors} \label{section performance}

We have introduced three different learning-based methods for computing an SACBF. However, learning errors are unavoidable and can arise from a combination of factors including insufficient data, loss function mismatch, and suboptimal optimization. These errors have the potential to invalidate the learned SACBF. To address this problem, we propose a systematic analysis to evaluate the robust safety performance of \eqref{DDfilter} when it includes an inaccurate approximation $Q^B_\omega$. This analysis further leads to a practical approach for handling learning errors through state constraint tightening followed by SACBF constraint relaxation.

Our analysis is inspired by the concept of input-to-state safety (ISSf) \cite{kolathaya2018input}, which was originally developed for studying set invariance in the presence of disturbances. 
Before introducing the framework, we point out that the analysis of safety performance can be applied to the SL method in Section \ref{section supervised} and the expert-guided learning method in Section \ref{section expert}, while the RL method in Section \ref{section RL SACBF} is excluded. We will explain the reason for this exclusion at the end of this section. 

First, we relax the constraint in the original optimization-based controller \eqref{DDfilter}, based on the intuition that it may no
longer be possible to render \eqref{DDfilter} recursively feasible in the presence of learning errors. Define a control policy with a relaxed SACBF constraint as
\begin{align}\label{DDfilter_relaxed}
	\pi^\mathrm{r}_{\theta,\omega}(x) :=\arg\min_{u\in U}&\; Q_\theta(x,u) \nonumber\\
	\mathrm{s.t.}\;&  Q^B_\omega(x,u)\leq \kappa(\varepsilon),
\end{align}
where $\kappa(\cdot)$ is a $\mathcal{K}_\infty$ function and $\varepsilon\geq 0$. The notation $\varepsilon$ will be used to quantify the learning error, and $\kappa$ will be elucidated later.

Then, similar to how CBFs are extended to ISSf CBFs \cite{kolathaya2018input}, we introduce error-to-state safety SACBFs (ESSf SACBFs), capturing the set invariance when $Q^B_\omega$ is different from $Q^B$. 

%
%
%

\begin{definition}[$\varepsilon$-ESSf SACBF]\label{ESSf SACBF definition}
	A function $\hat Q^B(\cdot,\cdot) : \mathcal{X} \times U \to \mathbb{R}$ is called an \emph{$\varepsilon$-error-to-state safe SACBF} ($\varepsilon$-ESSf SACBF) with a corresponding safe set $\hat{\mathcal{S}}_Q$, if there exists a $\mathcal{K}_\infty$ function $\kappa(\cdot)$ such that the pair $(\hat Q^B, \; \hat{\mathcal{S}}_Q)$ satisfies the following conditions:
	
		\noindent(i) $\hat{\mathcal{S}}_Q$ is non-empty. 
	
		\noindent(ii) $\min_{u\in U} \hat Q^B(x,u) \leq \kappa(\varepsilon) ,\;\forall x\in \hat{\mathcal{S}}_Q$.
	
		\noindent(iii) For any $x \in \hat{\mathcal{S}}_Q$, any $u \in U$ satisfying $\hat Q^B(x,u) \leq \kappa(\varepsilon) $ ensures that $f(x,u) \in \hat{\mathcal{S}}_Q$.
\end{definition}
From Definitions \ref{SACBF definition} and \ref{ESSf SACBF definition}, it is clear that if $\hat Q^B$ is an $\varepsilon$-ESSf SACBF, then $\hat Q^B - \kappa(\varepsilon)$ satisfies (i)-(ii) of Definition \ref{SACBF definition}. By further assuming that $\hat{\mathcal{S}}_Q$ is a subset of $X$, we will get the safety of the closed-loop system controlled by \eqref{DDfilter_relaxed}. Formally, we have the following results on the ESSf property of $Q^B_\omega$. 

Now, we are ready to give our main results on the ESSf property of $Q^B_\omega$ learned by the two methods presented in Sections \ref{section supervised} and \ref{section expert}.

\begin{theorem}[ESSf for the SL method]\label{performance theorem1}
	Under Assumption \ref{A1}, consider the controlled system $x_{t+1} = f(x_t, \pi^\mathrm{r}_{\theta,\omega}(x_t) )$, an SACBF $Q^B$, which is induced by a CBF $B$ from \eqref{relation} with $\beta_B<1$, and an approximation $Q^B_\omega$ of $Q^B$. Suppose that 
	\begin{equation}\label{error bound1}
				|Q^B(x,u) - Q^B_\omega(x,u)| \leq \varepsilon,\;\forall (x,u) \in \mathcal{X}\times U
    \end{equation}
	 holds. Letting $\kappa(\varepsilon)= \frac{1+\beta_B}{1-\beta_B} \varepsilon$, we have:
	
		\noindent(i) The approximation $Q^B_\omega$ is an $\varepsilon$-ESSf SACBF. The corresponding safe set is $\mathcal{S}_\omega =  \{x \in \mathcal{X}| B(x) \leq \frac{2 \beta_B }{1-\beta_B} \varepsilon \}$.
	
		\noindent(ii) The optimization-based controller \eqref{DDfilter_relaxed} is recursively feasible with the initial condition $x \in \mathcal{S}_\omega$, i.e., $\mathcal{S}_\omega$ is forward invariant for the controlled system $x_{t+1} = f(x_t, \pi^\mathrm{r}_{\theta,\omega}(x_t) )$.
	
		\noindent(iii) Furthermore, if $B$ is the CBF for \eqref{system} under the \emph{tightened} state constraint $ x \in X_\varepsilon := \{x \in \mathbb{R}^n | h(x) + \frac{2 \beta_B }{1-\beta_B} \varepsilon  \leq 0\}$, the controlled system satisfies the original constraints $x_t \in X$ and $u_t \in U$, $\forall t\in \mathbb{N}$.
\end{theorem}
\begin{proof}
	According to Lemma 1, $\mathcal{S}_Q = \mathcal{S}_B$. The non-emptiness of $\mathcal{S}_\omega$ follows directly from the non-emptiness of $\mathcal{S}_B$. Then, we will prove the feasibility of \eqref{DDfilter_relaxed} for any $x \in \mathcal{S}_\omega$ (Condition (ii) of Definition \ref{ESSf SACBF definition}). From \eqref{relation} and using the properties of CBFs, we have $\min_{u \in U} Q(x,u) \leq \max\{h(x),\;B(x)\} \leq B(x),\;\forall x \in \mathcal{X}$. For all $x \in \mathcal{S}_\omega$, it follows from \eqref{error bound1} that 
	\begin{align*}
		\min_{u \in U} Q^B_\omega(x,u) &\leq \min_{u \in U} Q^B(x,u) + \varepsilon\\
		&\leq B(x) + \varepsilon \leq \underbrace {\frac{{1 + {\beta _B}}}{{1 - {\beta _B}}}\varepsilon }_{\kappa(\varepsilon)}.
	\end{align*}
	
	Next, let $x \in \mathcal{S}_\omega$ and $u$ be any input satisfying $u\in U$ and $Q^B_\omega(x,u) \leq \kappa(\varepsilon)$. We have 
	\begin{align*}
		B(f(x,u)) \leq \beta_B Q^B(x,u) \leq  \beta_B Q^B_\omega(x,u) + \beta_B \varepsilon \leq \frac{2 \beta_B }{1-\beta_B} \varepsilon.
	\end{align*}
	The above inequalities prove the forward invariance of the set $\mathcal{S}_\omega$ and the satisfaction of Condition (iii) of Definition \ref{ESSf SACBF definition}. The above arguments prove the statements (i) and (ii).
	
	Furthermore, if $B$ is a CBF under the \emph{tightened} constraint $ x \in X_\varepsilon$, we derive that $B(x) - \frac{2 \beta_B }{1-\beta_B} \varepsilon \geq h(x)$, which further implies that the safe set $\mathcal{S}_\omega$ of the learned SACBF $Q^B_\omega$ is contained in the original state constraint set $X$. As $\mathcal{S}_\omega$ is forward invariant, the infinite-time safety of the controlled system follows. 
\end{proof}
\begin{theorem}[ESSf for the expert-guided learning method]\label{performance theorem2}
	Under Assumption \ref{A2}, consider the controlled system $x_{t+1} = f(x_t, \pi^\mathrm{r}_{\theta,\omega}(x_t) )$ and approximations $Q^B_\omega$ and $q_\omega$ of $Q^B$ and $q$ in \eqref{robust optimization}. Suppose that 
	\begin{subequations}\label{condition2}
		\begin{align}
			&Q^B_\omega(x,\pi_\mathrm{s}(x)) \leq \beta q_\omega(x) + \varepsilon,\;\forall x \in \mathcal{S}_\omega \label{error bound2 a} \\
			&q_\omega(f(x,u)) \leq Q^B_\omega(x,u) + \varepsilon,\;\forall (x,u) \in \mathcal{S}_\omega\times U \label{error bound2 b} \\
			& \min_{x\in \mathcal{X}} q_\omega (x) \leq  \frac{2  }{1-\beta} \varepsilon \label{error bound2 c}
		\end{align}
	\end{subequations}
	 holds. Letting $\kappa(\varepsilon)= \frac{1+\beta }{1-\beta} \varepsilon$, we have:
	
		\noindent(i) The approximation $Q^B_\omega$ is an $\varepsilon$-ESSf SACBF. The corresponding safe set is $\mathcal{S}_\omega =  \{x \in \mathbb{R}^n | q_\omega (x)  \leq  \frac{2  }{1-\beta} \varepsilon \}$.
	
		\noindent(ii) The optimization-based controller \eqref{DDfilter_relaxed} is recursively feasible with the initial condition $x \in \mathcal{S}_\omega$, i.e., $\mathcal{S}_\omega$ is forward invariant for the controlled system $x_{t+1} = f(x_t, \pi^\mathrm{r}_{\theta,\omega}(x_t) )$.
	
		\noindent(iii) Furthermore, if $h(x) + \frac{2 }{1-\beta} \varepsilon  \leq q_\omega(x),\;\forall x \in \mathcal{X} $, the controlled system satisfies the original constraints $x_t \in X$ and $u_t \in U$, $\forall t\in \mathbb{N}$.
\end{theorem}
\begin{proof}
	 Consider any $x\in  \mathcal{S}_\omega$, we have
		\begin{align*}
			\min _{u \in U} Q_\omega^B(x, u) \leq Q_\omega^B(x, \pi_\mathrm{s}(x))\leq \beta q_\omega(x) +\varepsilon  \leq \underbrace{\frac{1+\beta}{1-\beta} \varepsilon}_{\kappa(\varepsilon)}, 
		\end{align*}
	which proves the feasibility of \eqref{DDfilter_relaxed}. 
	
	Moreover, consider any $x \in \mathcal{S}_\omega$ and any $u\in U$ satisfying $Q^B_\omega(x,u) \leq \kappa(\varepsilon)$. We have 
	\begin{align*}
		q_\omega (f(x,u)) \leq Q^B_\omega(x,u) + \varepsilon \leq \frac{2}{1-\beta} \varepsilon.
	\end{align*}
     The rest of the proof can be completed using the same reasoning as that of the proof of Theorem \ref{performance theorem1}.
\end{proof}

	Theorems \ref{performance theorem1} and \ref{performance theorem2} provide a systematic way to explicitly quantify sufficient amounts of constraint tightening and SACBF relaxation required to retain safety for the controller $\pi^\mathrm{r}_{\theta,\omega}$ obtained from the relaxed optimization problem \eqref{DDfilter_relaxed}. In particular, Theorem \ref{performance theorem1} states that if the approximation $Q^B_\omega$ satisfies condition \eqref{error bound1}, to ensure safety for the SL method, the state constraint must be tightened to $ h(x) + \frac{2 \beta_B }{1-\beta_B} \varepsilon  \leq 0$, while the safety constraint in the safety filter should be relaxed to $ Q^B_\omega (x,u) \leq \frac{1+\beta_B}{1-\beta_B} \varepsilon$. Similarly, in the expert-guided learning method, Theorem \ref{performance theorem2} tells that the state constraint should be tightened to $ h(x) + \frac{2 }{1-\beta} \varepsilon  \leq 0$, with the relaxed SACBF constraint $\leftarrow Q^B_\omega (x,u) \leq \frac{1+\beta}{1-\beta} \varepsilon$ in the safety filter.

    The safety performance analysis in Theorems  \ref{performance theorem1} and \ref{performance theorem2} is only valid when $\beta$ is bounded substantially below 1. Otherwise, the tightened state constraint $X_\omega$ shrinks drastically and can even become empty. However, for the expert-guided learning approach, this limitation can be mitigated in practice since $\beta$  can be tuned through an outer-loop procedure to maintain a value sufficiently below 1.

\begin{remark}
    The soundness of Theorem \ref{performance theorem2} relies on the conditions in \eqref{condition2}. The two conditions in \eqref{error bound2 a} and \eqref{error bound2 b} are motivated by \eqref{robust optimization c} and \eqref{robust optimization d}, where we allow the violations of \eqref{robust optimization c} and \eqref{robust optimization d} to be bounded by $\varepsilon$. Besides, condition \eqref{error bound2 c} is introduced to ensure that $\mathcal{S}_\omega$ is non-empty. Condition \eqref{error bound2 c} can be checked by solving a nonlinear optimization problem, while checking the other conditions and computing $\varepsilon$ can be achieved by using existing deterministic or probabilistic verification tools \cite{he2024approximate,robey2020learning}.
\end{remark}
\begin{remark}
	Classical ISSf \cite{kolathaya2018input} or MPC-based tightening techniques \cite{hertneck2018learning} operate in the state space and rely on explicit knowledge of system dynamics. In contrast, our ESSf framework is built upon SACBFs, which evaluate safety directly on $(x,u)$ pairs. This fundamental shift eliminates the need for model predictions and allows safety analysis entirely from data. To the best of our knowledge, such an error-to-state formulation for state-action certificates in a direct data-driven setting has not appeared in prior work.

 Compared to robust MPC \cite{hertneck2018learning}, which typically relies on constraint tightening, ensuring both safety and recursive feasibility for our optimization-based controller is more challenging. In MPC, recursive feasibility directly guarantees constraint satisfaction of the resulting policy, but this implication does not generally hold for the proposed controller. To address this challenge, we introduce a novel constraint tightening-relaxation approach. Our analysis provides closed-form, quantitative relationships between learning errors in the SACBF, the required tightening of state constraints, and the necessary relaxation of the SACBF constraint.
\end{remark}

\begin{remark}
    \hk{Existing robust and observer-based CBF approaches can provide safety guarantees for \emph{partially unknown} systems, under assumptions on bounded disturbances or bounded disturbance-estimation errors \cite{westenbroek2021combining,isaly2024adaptive}. In these approaches, the CBF condition is still evaluated through a model-based representation, typically consisting of a nominal model together with a bounded uncertainty or estimation-error description. In contrast, we consider \emph{fully unknown} systems (without a nominal model). Instead, our robustness condition is expressed in terms of the SACBF approximation error.}
\end{remark}

We explain why the RL method cannot be included in the proposed ESSf framework. To make $Q^{B*}$ in \eqref{Q safe function} an SACBF, the optimal control problem \eqref{CBF generator} should be undiscounted, which leads to a non-contractive Bellman equation \eqref{bellman} \cite{bertsekas2022abstract}. As a consequence, the learning error of $Q^{B*}$ could become unbounded. Besides, in Theorem \ref{performance theorem1} we require $\beta_B<1$ while in the reachability formulation \eqref{CBF generator}, $\beta_B=1$. The above two factors make the ESSf analysis inapplicable to the RL method.

\section{Refining policies with SACBF constraints}

Given a valid SACBF obtained by the methods in the previous section, in this section, we update $\theta$ to refine the feasible policy $\pi_{\theta, \omega}$, bringing it closer to the optimal policy that solves \eqref{infinite}.

A convenient approach for designing $Q_\theta$, allowing $\theta$ to be updated using most existing unconstrained policy-based RL algorithms is to consider a Euclidean distance objective function $Q_\theta(x,u) := ||u-\pi_\theta(x)||_2$, where $\pi_\theta: \mathcal{X} \rightarrow \mathcal{U}$ is an explicit controller (e.g., a neural network controller) with the parameter $\theta$ \cite{he2024approximate}. The resulting controller is referred to as the direct data-driven safety filter (3DSF). This approach, however, could compromise the optimality of the projected policy $\pi_{\theta, \omega}$. A less conservative approach is to make $Q_\theta$ approximate the constrained optimal value function $Q^*$, defined by $Q^*(x,u):= g(x,u) + \gamma J^*(f(x,u))$. 

Normally, to obtain $Q^*$, as is shown in \cite[Chapter 7.4.1]{borrelli2017predictive}, one typically needs to recursively compute the backward reachable sets $X_k,\;k\in \mathbb{N}$ and enforce the constraint $f(x,u) \in X_k$ during Bellman iterations. This procedure is in general intractable for unknown nonlinear systems. Instead, we enforce the obtained SACBF constraints $Q_{\omega}^B(x_k, \pi(x_k))\leq 0,\;\forall k\in \mathbb{N}$ during Bellman iterations. Note that this will not introduce any conservatism if $Q_{\omega}^B = Q^{B*}$. The benefit is that the constraints during Bellman iterations become fixed and model-free. This is a result of the inherent invariance of the safe set of $Q_{\omega}^B$. 

By combining the following expression of $\bar J^*$:
\begin{align}\label{infinite under Q}
	\bar J^*\left(x\right) := \inf _{\pi}  J^\pi\left(x\right)\;\mathrm{s.t.}\; Q_{\omega}^B(x_t, \pi(x_t))\leq 0, \pi(x_t) \in U, t\in \mathbb{N},
\end{align}
and the expression \eqref{DDfilter} of $\pi_{\theta, \omega}$, we know that $\pi_{\theta, \omega}$ is optimal for the \eqref{infinite under Q} if $Q_\theta$ equals $\bar Q^* (x, u) := g(x, u)+\gamma \bar J^*(f(x, u))$.

The constrained optimal state-action value function $\bar Q^*$ satisfies the following constrained Bellman equation:
\begin{align}\label{bellman2}
	\bar Q^*(x,u) =\Gamma \bar Q^*:= g(x,u) + \gamma &\min_{u^+ \in U}\bar Q^*(f(x,u),u^+) \nonumber\\
	&\mathrm{s.t.}\; Q_{\omega}^B(f(x,u), u^+) \leq 0,
\end{align}
where $\Gamma $ is the Bellman operator. It is easy to show that $\Gamma $ is a monotonous contraction mapping. As a consequence, the uniqueness of the solution to \eqref{bellman2} holds and the constrained Bellman iteration $\bar Q_{k+1} = \Gamma \bar Q_k$ converges to $\bar Q^*$ for any real-valued and bounded $\bar Q_0: \mathcal{X}\times \mathcal{U} \to \mathbb{R}$ \cite[Proposition 2.1.1]{bertsekas2022abstract}.

With this property, we can update $\theta$ to minimize the average Bellman residual over $\mathcal{S}_{\omega} \times U$. Formally speaking, we find
\begin{subequations}
	\begin{align}\label{learning objective}
		\theta^* := \arg\min_{\theta} \int_{\mathcal{S}_{\omega} \times U} (  
		q_\theta(x,u) &- Q_\theta(x,u) )^2 dxdu,\\
		\text{where }q_\theta(x,u) =	g(x,u) + \gamma 
		& \min_{u^{+} \in U} Q_\theta\left(f(x, u), u^{+}\right) \nonumber \\
		& \text {s.t. } Q_\omega^B\left(f(x, u), u^{+}\right) \leq 0.
	\end{align}
\end{subequations}
Although $f$ is used in the Bellman notation, in implementation it is replaced by the observed successor $x^+$. Problem \eqref{learning objective} is bi-level and intractable to solve exactly. In practice, following standard offline value iteration algorithms \cite{busoniu2017reinforcement}, we iteratively update $\theta$. In each iteration, $q_\theta$ is treated as the target value, and $\theta$ is updated such that the difference between $ q_\theta$ and $Q_\theta$ is minimized. This ultimately forms the proposed constrained fitted Q iteration (constrained FQI) algorithm (Algorithm 1). Furthermore, just as standard FQI can be adapted to constrained FQI, other online value-based RL algorithms, such as constrained deep Q-learning, constrained approximate SARSA, and Lagrangian RL, can likewise be designed accordingly.


\begin{algorithm}
	\caption{Constrained fitted Q iteration}
	\label{alg:C}
	\begin{algorithmic}[1]
		
		\State \textbf{Given} the SACBF $Q^B_\omega$, the state and input constraint sets $X$ and $U$, $\mathcal{X}$, the sample set $\mathcal{D} \subseteq \mathcal{S}_{\omega}\times U\times \mathcal{S}_{\omega}$, the error threshold $\varepsilon_{\mathrm{QI}} >0$, and the maximum number of updates $K \in \mathbb{N}^+$
		
		\State \textbf{Initialize} $\theta_0$, $k\leftarrow0$

		\State\textbf{Repeat} at each iteration $k$
		
		$q_s \leftarrow  \Gamma Q_{\theta_k} (x_s,u_s)$ for each $(x_s,u_s,x^+_s) \in \mathcal{D}$
		
		$\theta_{k+1} \leftarrow  \arg \min _\theta \sum_{s=1}^{|\mathcal{D}|}\left(q_s-Q_{\theta}\left(x_{s}, u_{s}\right)\right)^2$

		$k \leftarrow k+1$
		
		\State\textbf{Until} $|Q_{\theta_{k+1}}(x,u) -Q_{\theta_{k}}(x,u)| \leq \varepsilon_{\mathrm{QI}},\;\forall (x,u) \in \mathcal{X}\times U $,  or $k > K$
		\State \textbf{Output} $Q_{\theta_k}$
	\end{algorithmic}
\end{algorithm} 

Under the additional assumption that for any $k \in \mathbb{N}$, the deviation $|Q_{\theta_{k+1}}(x,u) - \Gamma Q_{\theta_k}(x,u)|$ is uniformly upper-bounded by $\delta$, a standard condition in RL literature (e.g., \cite[Section 2.3.1]{bertsekas2022abstract}), it follows that $\lim\sup_{k\to \infty} ||Q_{\theta_k}-\bar Q^*|| \leq\delta/(1-\gamma) $ \cite[Proposition 2.3.2]{bertsekas2022abstract}, where $\bar Q^*$ is the unique fixed point of $\bar Q^*=\Gamma\bar Q^* $. Moreover, note that the update of $Q_\theta$ does not affect the guarantees of safety or recursive feasibility of \eqref{DDfilter_relaxed}.

\section{Discussion}
\noindent\textbf{Computational complexity:} In our work, we use two deep neural networks $Q_\theta$ and $Q^B_\omega$ to parameterize the objective function $Q$ and the constraint function $Q^B$ in (7), respectively. Let $Q_\theta$ have $L_Q$ layers with widths $(d_0,\dots,d_{L_Q})$, where
$d_0 = n_x+n_u$ and $d_{L_Q}=1$, and analogously $Q^B$ has widths $(\tilde d_0,\dots,\tilde d_{L_B})$. The number of parameters of the two networks is
\[
P_Q = \sum_{l=1}^{L_Q} (d_{l-1} d_l + d_l), 
\qquad
P_B = \sum_{l=1}^{L_B} (\tilde d_{l-1} \tilde d_l + \tilde d_l).
\]
A forward pass performs (approximately) one multiply and one add per weight, so the number of floating-point operations (FLOPs) per evaluation satisfies
\[
C_{\mathrm{for}} = \mathcal{O}(P_Q + P_B).
\]

Now assume that we solve the nonlinear program in the optimization-based controller by an iterative first-order method (e.g., projected gradient, penalty or barrier method). Let $K$ denote the number of iterations
per control step. Each iteration requires one forward and backward pass through both networks, and the cost of a backward pass is on the same order as the forward pass, i.e., $C_{\mathrm{back}}=C_{\mathrm{for}}$. In addition, some cheap vector operations are needed to handle the constraint $u\in U$ (e.g., projections, line search updates), which takes $\mathcal{O}(n_u)$ FLOPs. Usually, $P_Q + P_B \gg n_u$. Hence, over $K$ iterations, the total complexity per control step is
\[
C_{\mathrm{comp}}
= \mathcal{O}\bigl(K (P_Q + P_B)\bigr).
\]

\noindent\textbf{Memory usage:} We distinguish static (always present) and dynamic (per evaluation/iteration) memory. 

The static memory is the memory used for storing the networks:
\[
M_{\mathrm{sta}} = \mathcal{O}(P_Q + P_B).
\]

During a forward and backward pass, activations, gradients, and the final solution (control input) must be stored. Hence the total dynamic memory per iteration is
\[
M_{\mathrm{dyn}}
= \mathcal{O}\left(\sum_{l=1}^{L_Q} d_l + \sum_{l=1}^{L_B} \tilde d_l + n_u\right),
\]

Since $P_Q + P_B$ dominates $\sum_{l=1}^{L_Q} d_l + \sum_{l=1}^{L_B} \tilde d_l + n_u$, we can summarize that the total memory is 
\[
M_{\mathrm{total}}= \mathcal{O}(P_Q + P_B).
\]

\noindent\textbf{Comparison with learning CBFs:} Machine learning has recently been developed for synthesizing CBFs (safety filters) with or without model uncertainties. A common methodology to synthesizing permissive CBFs is to link CBFs with value functions of optimal control problems, such as the approaches in \cite{fisac2018general,chen2024learning,didier2024approximate,tan2023value} and the current paper, or value functions of known policies \cite{so2024train,zhang2025discrete}. Some authors assume the knowledge of a nominal model and a valid CBF for that model \cite{taylor2020learning,westenbroek2021combining}. Then, RL \cite{taylor2020learning,westenbroek2021combining} or Gaussian process learning \cite{castaneda2025recursively} is used to learn the discrepancies between the CBF constraint associated with the nominal model and the real model. By leveraging the covariance of the estimation provided by the Gaussian process and developing an event-trigger mechanism, \cite{castaneda2025recursively} ensures the recursive feasibility of the safety filter in high probability \cite{castaneda2025recursively}. Another approach is to formulate the synthesis of CBFs as a semi-infinite optimization problem and to use sampling-based methods to solve it \cite{robey2020learning,nejati2023data}. This approach is more sample-efficient than RL. However, all the above methods can only enable the design of model-based safety filters. In contrast, in the current paper we fundamentally propose a new direct data-driven safe control framework, in which an SACBF is employed to evaluate the feasibility of input signals without using model information.

The paper in \cite{lavanakul2024safety} uses discriminating hyperplanes (a series of linear constraints) to represent safe constraints in safety filters. This approach eliminates the dependence on any specific safety certificate, and consequently, on the model. The hyperplanes can be updated by RL, in which the reward is designed such that unsafe policies are highly penalized. However, the RL method of \cite{lavanakul2024safety} can still be understood as a reward shaping method for managing safety, which lacks formal guarantees regarding constraint satisfaction. Besides, approximating nonlinear constraints by linear constraints may lead to a rather conservative policy. In contrast, the proposed optimization-based control framework, which involves a nonlinear program, provides a formally sound and more general method to enhance safety and to eliminate dependence on an explicit model.

The three approaches we proposed can be viewed as extensions of traditional CBF synthesis methods (typically model-based) \cite{robey2020learning,wabersich2022predictive,choi2021robust} to the synthesis of SACBFs in a model-free framework. Since SACBFs remove the model dependency by enforcing the invariance condition directly on the joint state-action space, learning SACBFs requires learning a safety function over a higher-dimensional domain with dramatically more samples and without relying on an explicit model. Secondly, since the safe set is not explicitly defined by the level set, SACBF learning must jointly infer the safe set,
the SACBF function, and the invariance structure. Finally, ESSf robustness analysis is more involved for SACBFs because it must characterize the effect of learning errors without propagating them through the system model. For these reasons, learning SACBFs involves fundamentally different technical challenges compared to learning CBFs.

\noindent\textbf{Comparison with integrating RL and MPC:}  The combination of MPC and RL can have various kinds of forms. \cite{gros2022learning} uses RL algorithms to update a parameterized nonlinear MPC scheme. It is shown in \cite{gros2022learning} that the parameterized MPC scheme can produce safe and stabilizing policies as long as the RL algorithm updates parameters in a personally-defined safe and stable set. Such a set, however, is usually problem-dependent. It is possible to use LMIs to determine such a set for linear systems. Different from \cite{gros2022learning}, which parameterizes all terms of MPC including the terminal cost, the prediction model, and the constraints, \cite{moreno2022predictive} only parameterizes the terminal cost and uses approximate value iteration to learn it offline. Similarly, \cite{lin2023reinforcement} adopts approximate policy iteration to learn the terminal cost offline. It is proven in \cite{moreno2022predictive,lin2023reinforcement} that the resulting MPC controller makes the system safe and asymptotically stable if the approximation error is bounded and the MPC horizon is large enough. The above combinations, however, fail to solve the online computational problem faced by MPC, since their policies are still determined by online optimization over a long prediction horizon. In comparison, the proposed optimization-based control approach yields significant online computational benefits by using two state-action value functions $Q_\theta$ and $Q^B_\omega$ to approximate the value function and safety constraints implicitly defined by the parameterized MPC scheme. 

Moreover, we acknowledge that, unlike parametrized MPC, which employs an explicit prediction model to derive optimal and safe policies for long-term goals, our optimization-based control framework lacks explainability in the resulting policy. However, our method offers greater flexibility, as it allows the integration of any RL algorithm into the control synthesis.

\noindent\textbf{Comparison with safe RL:} Existing methods for safe RL include using reward shaping \cite{massiani2022safe,he2024approximate}, Lagrangian methods combined with policy gradient or actor-critic algorithms \cite{yu2022reachability}, interior-point optimization \cite{liu2020ipo}, and safety filtering \cite{cheng2019end,dalal2018safe}. Additionally, a small body of work explores direct stochastic optimization of neural network controllers over finite horizons, improving safety at sampled states but demanding significant computational resources \cite{li2022using}. The work of \cite{fisac2019bridging} approximates the optimal safe policy by using RL to solve a time-discounted modification of the Hamilton-Jacobi reachability problem, while \cite{hsu2023isaacs} extends this framework to adversarial RL, making the resulting policy robust to model uncertainty. More recently, \cite{li2025certifiable} establishes the Lipschitz continuity of the discounted value function introduced in \cite{fisac2019bridging} and leverages this property to develop certification methods for the corresponding approximate policy. However, as discussed in the introduction, because learning algorithms operate stochastically, learning-based control—particularly when an explicit model is unavailable—requires safety filters to regulate policy execution and to ensure reliable safety. The proposed optimization-based control framework is compatible with all the safe RL methods mentioned above, and provides formal safety assurance under learning errors. 

\noindent\textbf{Comparison with existing work on SACBF-based safety:} In \cite{he2023state}, a general definition of SACBFs is introduced. In parallel, the authors of \cite{fisac2019bridging} and \cite{oh2025safety} study the synthesis of SACBFs through model-free learning. The current work extends \cite{he2023state,fisac2019bridging,oh2025safety} in three major ways. First, we propose learning-based methods for synthesizing SACBFs without requiring any nominal model, which is assumed in \cite{he2023state}. The SACBF considered in this paper is not limited to the state-action value functions in \cite{he2023state,fisac2019bridging,oh2025safety}. Second, to address the effect of learning errors, this paper establishes robustness analysis via the notion of ESSf, resulting in significantly less conservative constraint tightening than \cite{he2023state}. In contrast, \cite{fisac2019bridging} and \cite{oh2025safety} do not examine the effect of learning errors. Third, unlike the safety-filter approaches in \cite{he2023state,fisac2019bridging,oh2025safety}, our optimization-based controller uses a general nonlinear objective $Q$ and a policy refinement strategy to optimize performance while preserving safety.

\noindent\textbf{Limitations:}
One downside is that safety constraints will inevitably be violated when collecting samples of state transitions to train the controller. This may be inappropriate for real-world applications where maintaining safety is essential throughout the learning process. Using simulators to generate sample data during learning can solve this issue, although the gap from sim to real needs further robustness analysis.

\hk{The extension of CBFs to SACBFs comes at the cost of increased computational and sample complexity, since the SACBF must be learned over the joint state-action space. Similar scalability challenges also arise in related RL methods \cite{robey2020learning,nejati2023data,he2024approximate}. Approaches for mitigating this complexity include exploiting system symmetries or using reduced-order representations. We leave the investigation of these directions to future work.}

We focus on discrete-time nonlinear systems because our work targets safety in learning-based control methods, like RL and SL, which typically use discrete-time models. The extension of SACBFs and the associated robustness analysis to continuous-time dynamical systems will be investigated in our future work.

\section{Case study}

We verify the proposed data-driven approaches for synthesizing safety filters and their application to safe learning-based control for an autonomous vehicle moving in a 2-D space containing obstacles.

\subsection{Model}
We consider the kinematic vehicle model \cite{didier2024approximate}:
\begin{equation}\label{car}
	\begin{aligned}
		\dot p_x & = v \cos (\Psi)  \\
		\dot p_y & = v \sin (\Psi) \\
		\dot v &= a \\
		\dot{\Psi} & =v \tan (\delta_\mathrm{s}) /L,\\
	\end{aligned}
\end{equation}
where $L= 0.1 $. The state vector $x$ includes the position $p_x$, $p_y$, the speed $v$, and the yaw angle $\Psi$. The acceleration $a$ and the steering angle $\delta_\mathrm{s}$ are the inputs. The input constraints are given by $-5\leq a \leq 2 $ and $|\delta_\mathrm{s}| \leq \pi/4 $. The state constraints are specified by $|p_x| \leq 2.6 $, $|p_y| \leq 2.6 $, $0\leq v \leq 1$, and $|\Psi| \leq \pi $, as well as the requirement of avoiding some obstacles shown in Fig. \ref{SACBF}. According to the considered state constraints, the function $h$ is as follows:
\begin{align*}
	h(x) =& \max \left\{ {|{p_x}| - 2.6,|{p_y}| - 2.6, - v,v - 1,|\Psi|-\pi},\begin{array}{*{20}{c}}
		{} \\ 
		{} 
	\end{array}\right.\\
	&\left.\max_{i=1,2,3,4}\{r^2_i - (p_x - c_{x,i})^2 - (p_y - c_{y,i})^2\} \right\}.
\end{align*}
Here, $r_i$ denotes the radius of each obstacle, while $c_{x,i}$ and $c_{y,i}$ represent the $x$- and $y$-coordinates of the center of each obstacle. The state constraint is tightened to $\{x| h(x)+ 0.2 \leq 0\}$ when learning SACBFs. The terminal set is $X_\mathrm{ter} = \{x\in \mathbb{R}^4|p^2_x + p^2_y \leq 0.1^2 \}$. The system is discretized using the Euler method with sampling time 0.05 s.

\subsection{Synthesizing SACBFs} The SL approach in Section \ref{section supervised} is not considered because it is unrealistic to assume knowing a CBF without knowing the model in this example. For the expert-guided approach, we adopt artificial potential fields (APF) \cite{warren1989global} to design the expert controller. \hk{The APF controller contains two elements: attractive potential fields, which guide the vehicle
toward the terminal set, and repulsive potential fields, which steer the vehicle away from the obstacles. As a result, the controller can generate constraint-satisfying trajectories that reach the terminal set for a nonempty set of initial conditions. The green trajectories in Fig. \ref{SACBF} verify Assumption \ref{A2}. The APF controller is analytically computed based on the gradients of attractive and repulsive potential functions.}

To get the training data sets $\mathcal{S}_\mathrm{s}$, $\mathcal{X}_\mathrm{s}$ for the expert-guided approach, and $U_\mathrm{s}$, we start the system from 10000 randomly generated initial conditions inside $\mathcal{X} = \{x\in \mathbb{R}^4 \mid\; |p_x| \leq 3,\;|p_y| \leq 3,\; -0.2\leq v \leq 1.2,\;|\Psi| \leq \pi\}$ and get the trajectories over 200 time steps. For the remaining learning models, including using RL to learn the SACBF $Q^B$, learning the optimal value function $Q$, learning the reference control policies, and identifying the model in the subsequent statistical tests, we use uniformly random sampling to get $10^6$ state-input samples from $\mathcal{X}$.

We use neural networks with [128 128 32] “tansig” layers to represent all the neural SACBF\footnote{The vector $[a_1 \;a_2\;...\;a_s]$ means the neural network contains $s$ hidden layers with $a_i$ units in each layer.}. The training algorithm is stochastic gradient descent with momentum \cite{ramezani2024generalization}, with a learning rate 0.001.  For the expert-guided learning approach, $q_\omega$ is represented by a neural network with [64 64 32] “tansig” layers, and the constraints in \eqref{robust optimization b}-\eqref{robust optimization d} are penalized in the loss with $\beta = 0.1$ and the penalty weights $\lambda_{\eqref{robust optimization b}} =1$ and $\lambda_{\eqref{robust optimization c}} =\lambda_{\eqref{robust optimization d}}=10$. After 10 epochs of training, it is verified that the three inequalities in \eqref{error bound2 a}-\eqref{error bound2 c} hold for $\varepsilon = 0.083$ at all the training samples. Therefore, following the probabilistic verification approach in \cite{hertneck2018learning}, we know that these three inequalities hold over the continuous sets $\mathcal{S}_0$, $\mathcal{S}_0 \times U$, and $\mathcal{X} \setminus \mathcal{S}_0$ with high confidence by employing Hoeffding’s Inequality \cite{hoeffding1994probability}. According to Theorem \ref{performance theorem2}, under the tightened constraint $h(x)+0.2 \leq 0$, the safe set $\{x| q_\omega (x) \leq 2\varepsilon/(1-\beta)\}$ is contained in the original state constraint with high probability.

\begin{figure}
    \centering
    \includegraphics[width=110pt,clip]{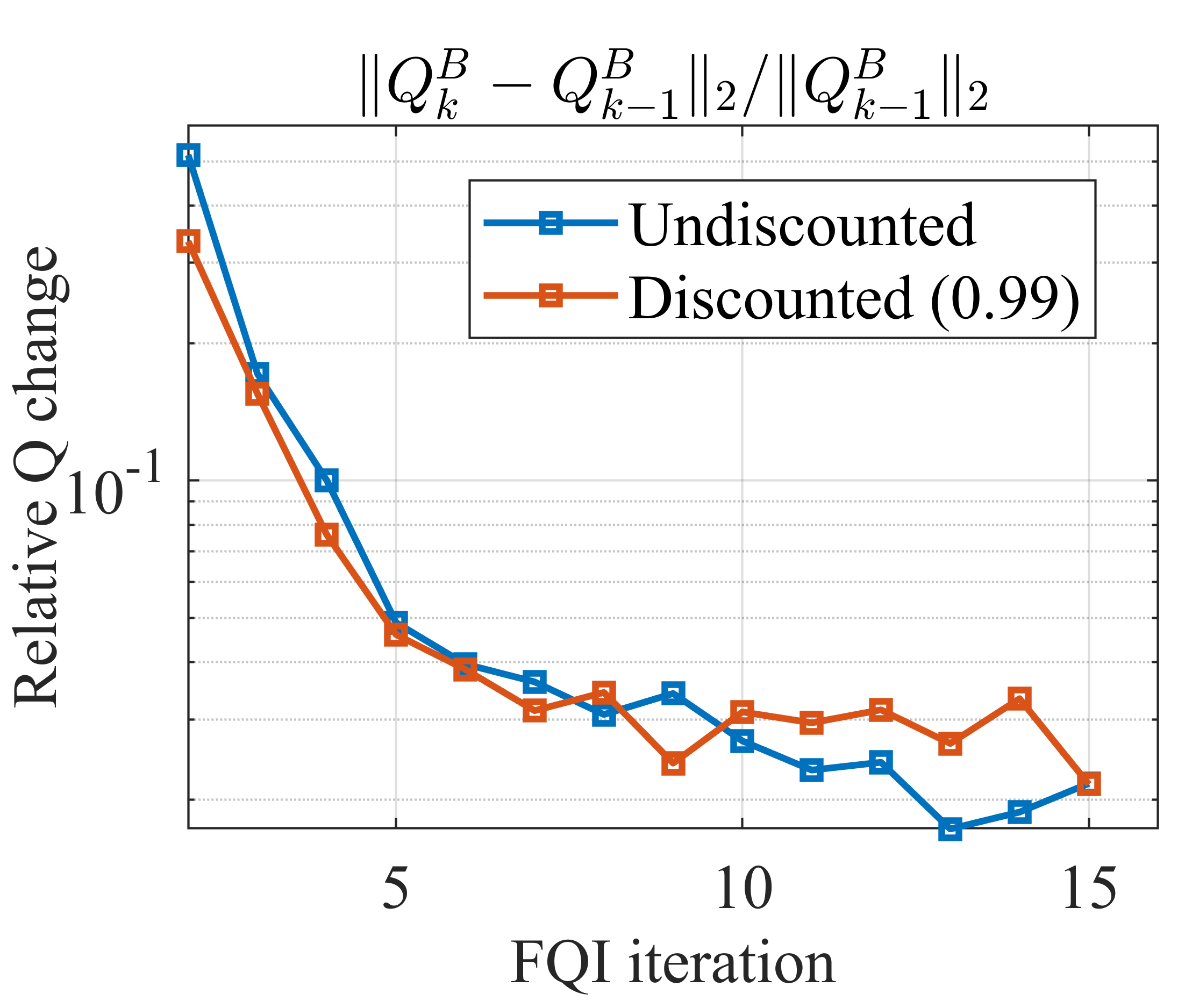}
	\hfil
	\includegraphics[width=110pt,clip]{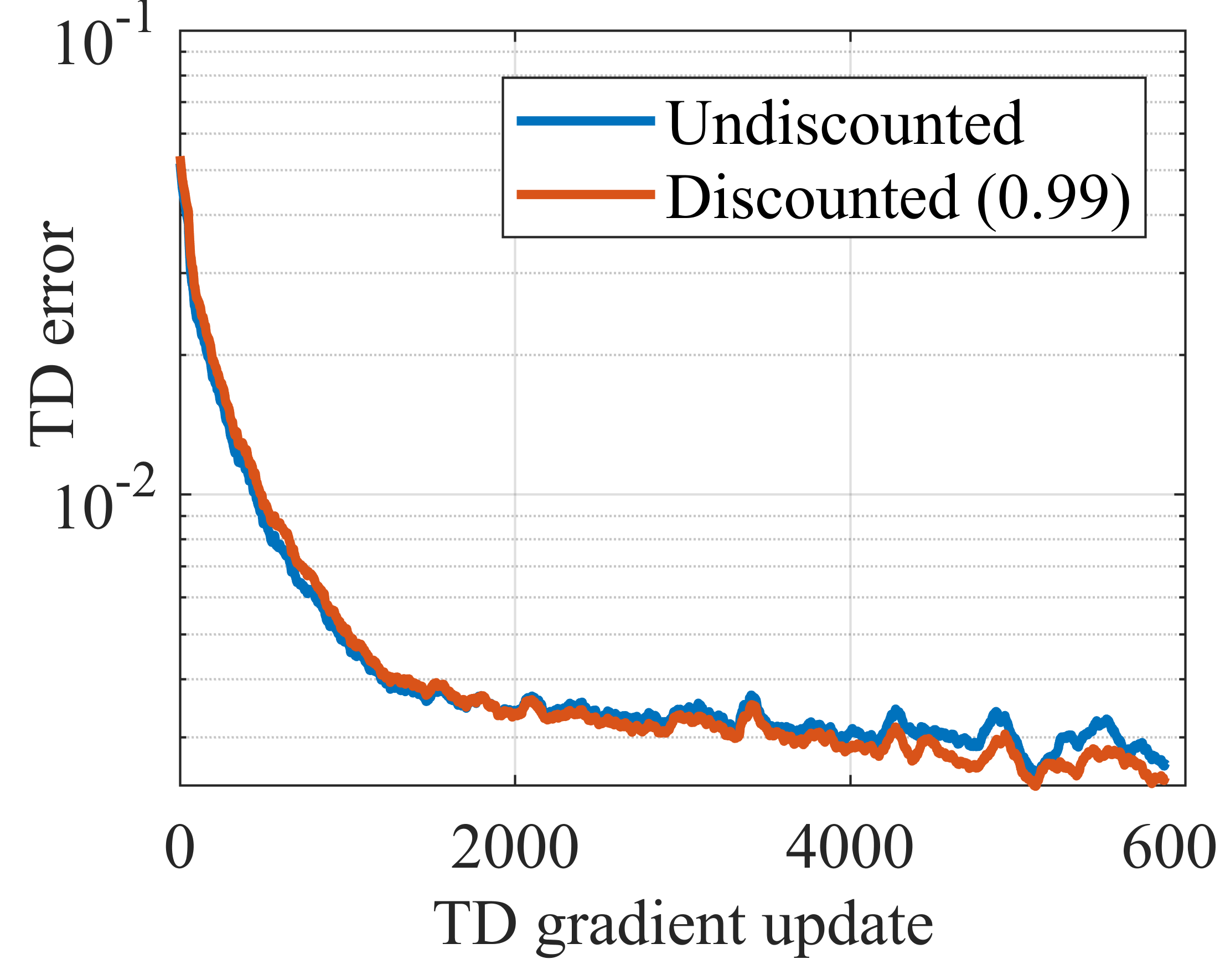}
	\caption{\textcolor{blue}{Empirical training behavior of neural FQI and temporal-difference learning for SACBF synthesis.}}
	\label{fig:convergence}
\end{figure}

\hk{Fig. \ref{fig:convergence} illustrates the empirical learning behavior of FQI and TD learning for SACBF synthesis. We compare the undiscounted reachability formulation \eqref{CBF generator} with the discounted formulation proposed in \cite{fisac2019bridging}. For FQI, we report the relative difference between the approximated SACBFs obtained in two consecutive iterations. For TD learning, we show the curve of the loss function \eqref{TDerror} during stochastic-gradient updates. Both the relative difference and the loss function decrease substantially and eventually plateau, indicating numerical stabilization of the learning algorithms in the case study.}

In Fig. \ref{SACBF}, we illustrate the neural SACBFs and their corresponding safe sets. To visualize the 4D sets, we display their shapes on the $p_x$-$p_y$ plane with fixed $v=0.5$ and $\Phi = 0$. The safe set derived from the expert controller is smaller than that learned through the RL approach. This difference arises because the target $Q^{B*}$ of the RL method corresponds to the maximal control-invariant set, whereas the expert controller (APF) is not always feasible within the maximal safe set, which is illustrated by the red trajectories in Fig. \ref{SACBF}. \hk{Fig. \ref{SACBF} (right) compares the safe sets learned by the FQI and TD algorithms under both the undiscounted and discounted formulations. All learned safe sets remain within the admissible state region. The safe sets obtained by TD learning are generally larger than those obtained by FQI. This is because an additional loss term $l_2$ is introduced in the TD-learning objective \eqref{TDerror} to bias the learned SACBF toward a smaller solution of the Bellman equation and, consequently, toward a larger safe set. In addition, the discounted and undiscounted formulations yield qualitatively similar safe-set boundaries.}

\begin{figure}
    \centering
    \includegraphics[width=110pt,clip]{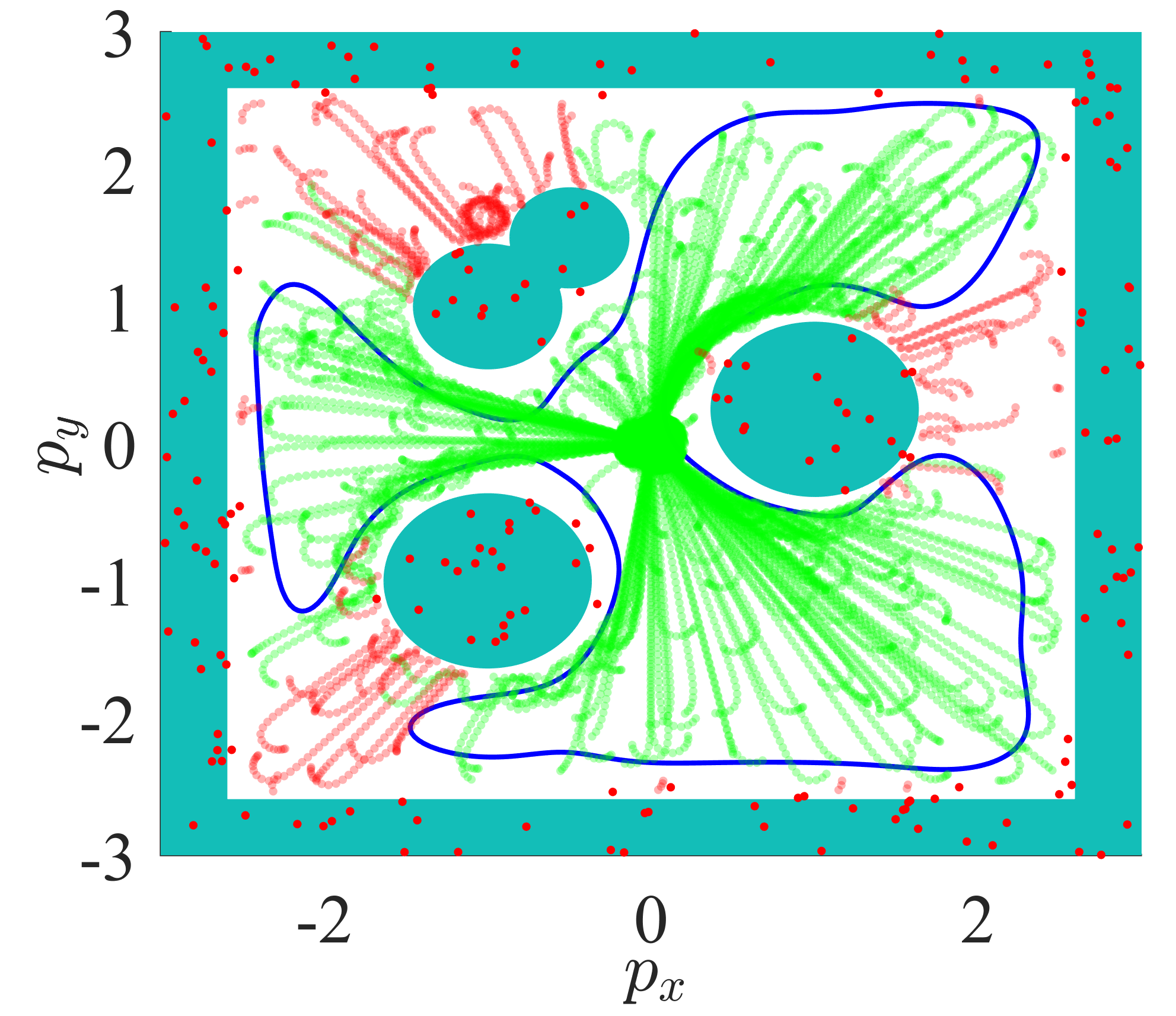}
    \hfil
    \includegraphics[width=110pt,clip]{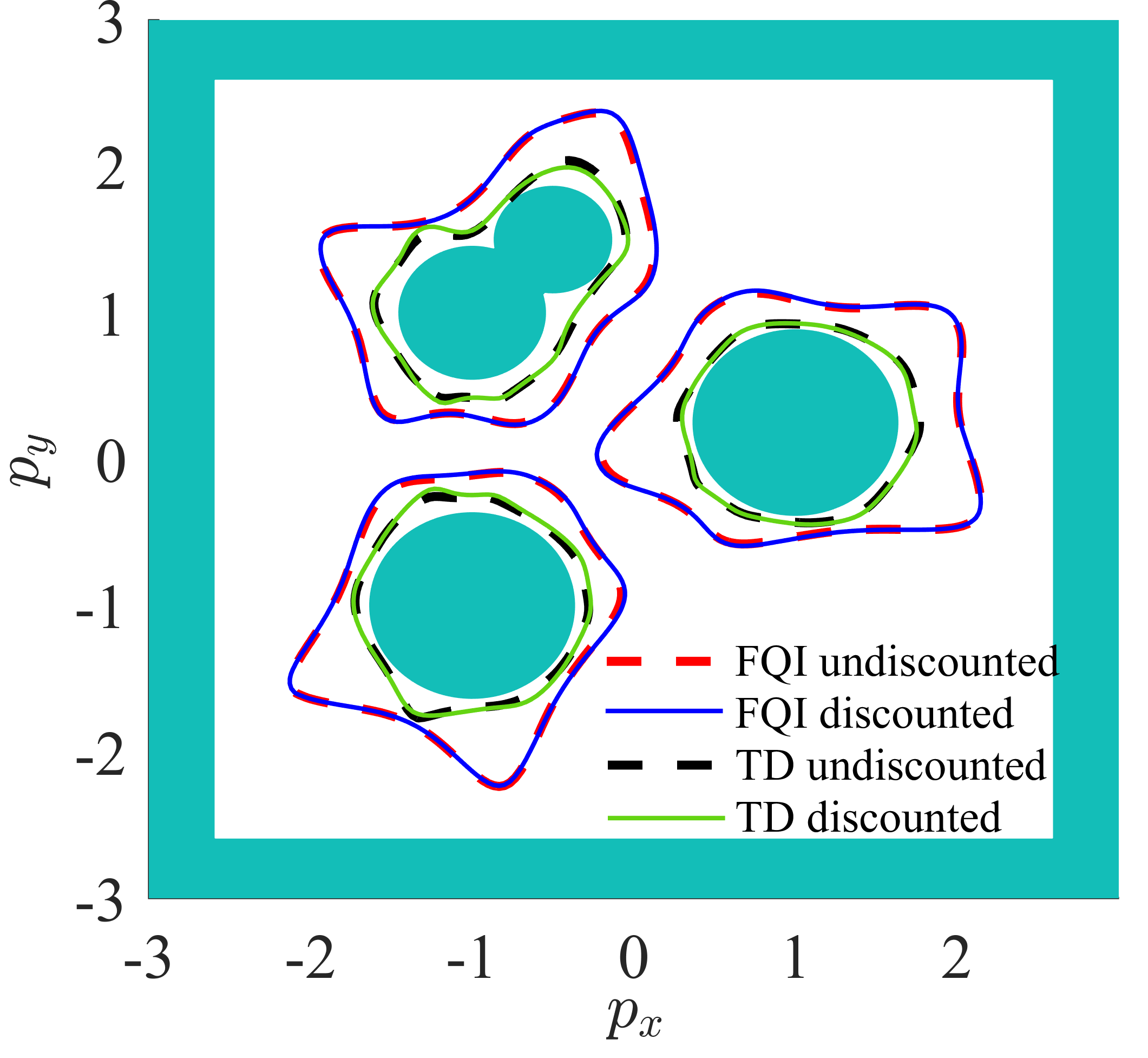}
	\caption{\hk{Boundaries (blue curves) of the safe sets learned by the expert-guided approach (left) and RL (right). The green curves represent the safe trajectories of the APF controller, while the red curves and points represent the unsafe trajectories and initial states.}}
	\label{SACBF}
\end{figure}

Fig. \ref{beta} demonstrates how the size of $S_\omega$ and $X_\omega$ varies with different values of $\beta$ when applying the expert-guiaded learning approach. From Fig. \ref{beta}, we can observe that the safe set is non-empty for all $\beta = 0,\;0.2,\;0.4,\;0.6,\;0.8$.

\begin{figure}
	\centering
	\subfloat[ ]{\includegraphics[width=120pt,clip]{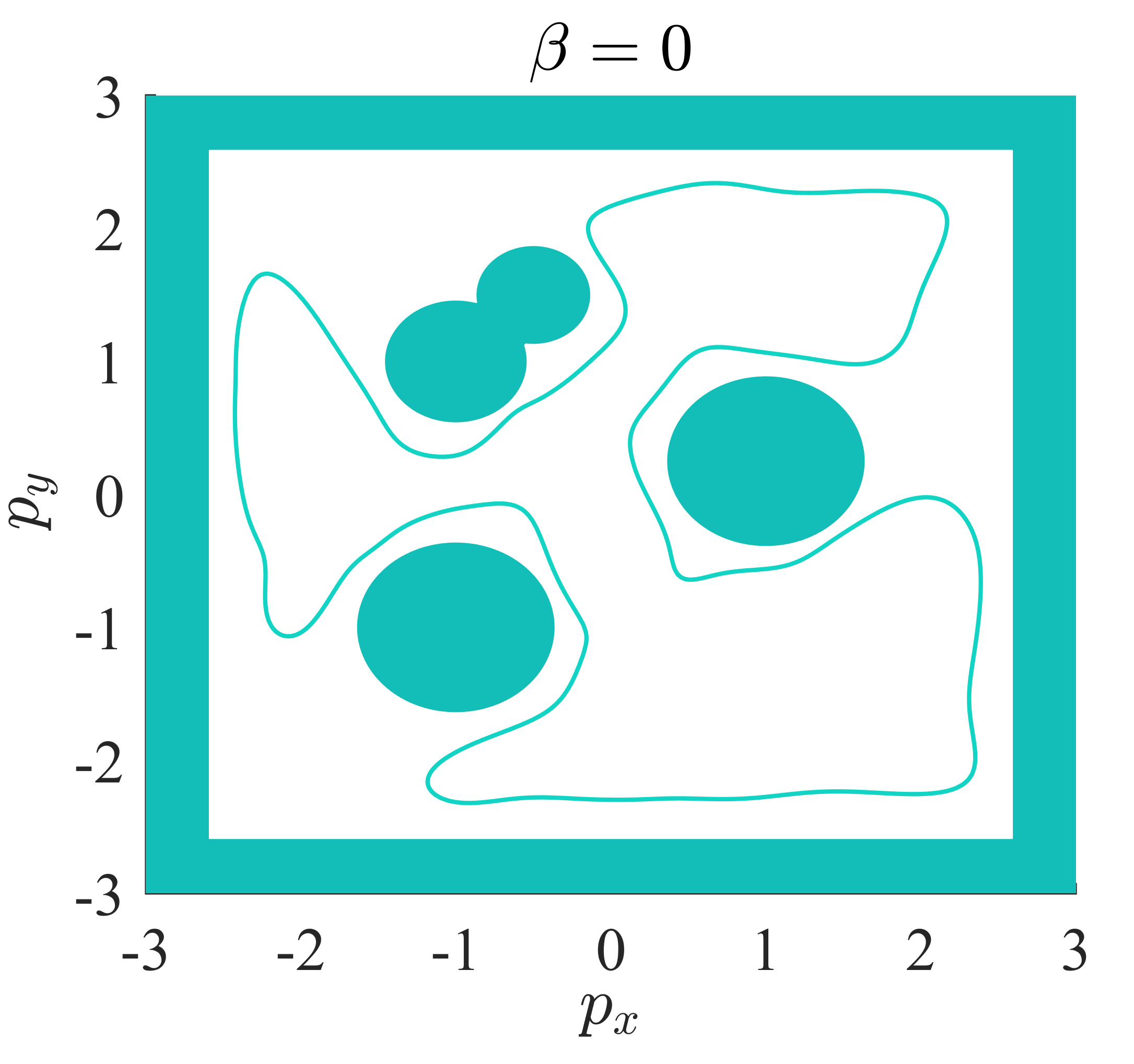}}
	\hfil
	\subfloat[]{\includegraphics[width=120pt,clip]{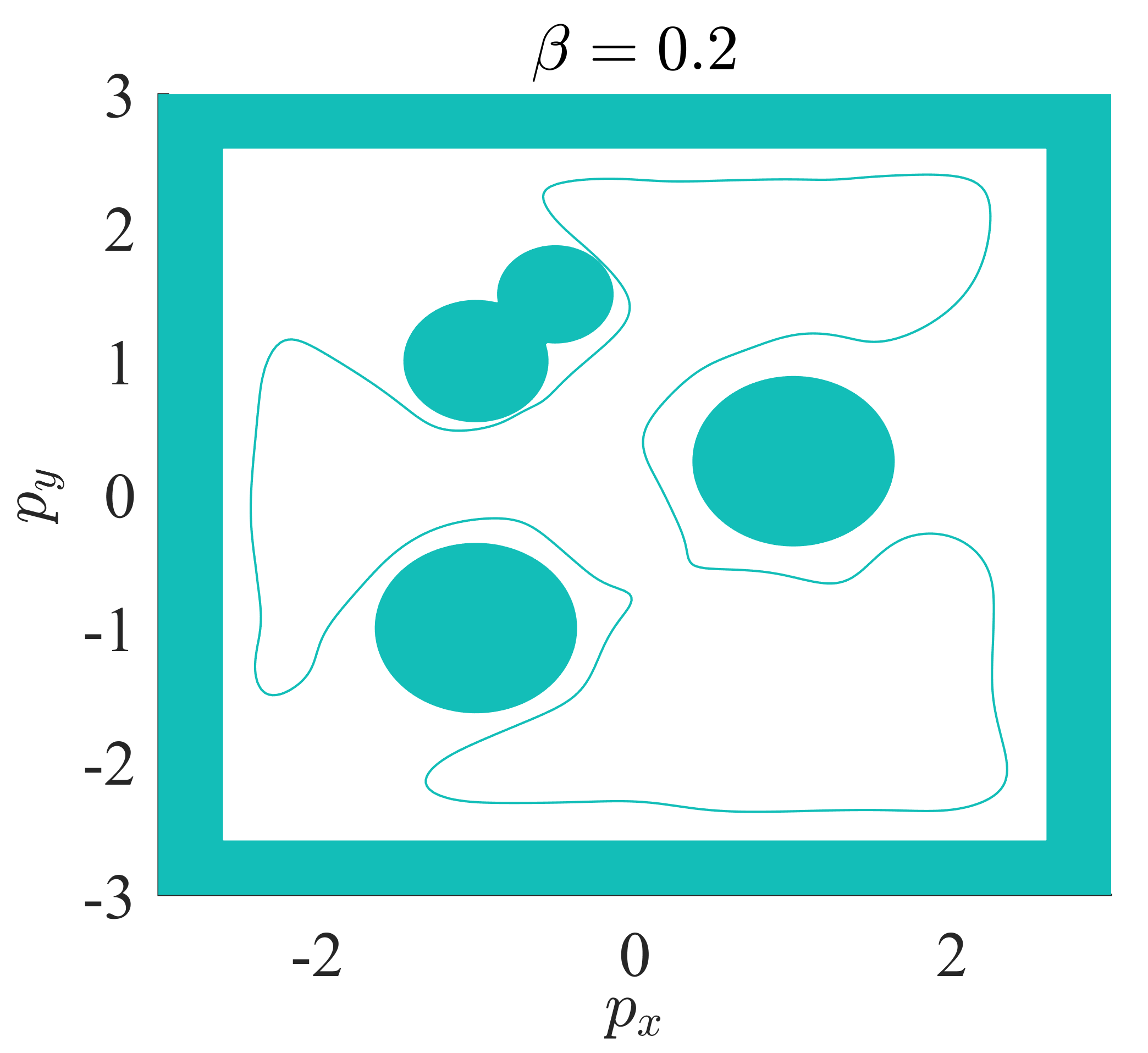}}
	\hfil\\
	\subfloat[]{\includegraphics[width=120pt,clip]{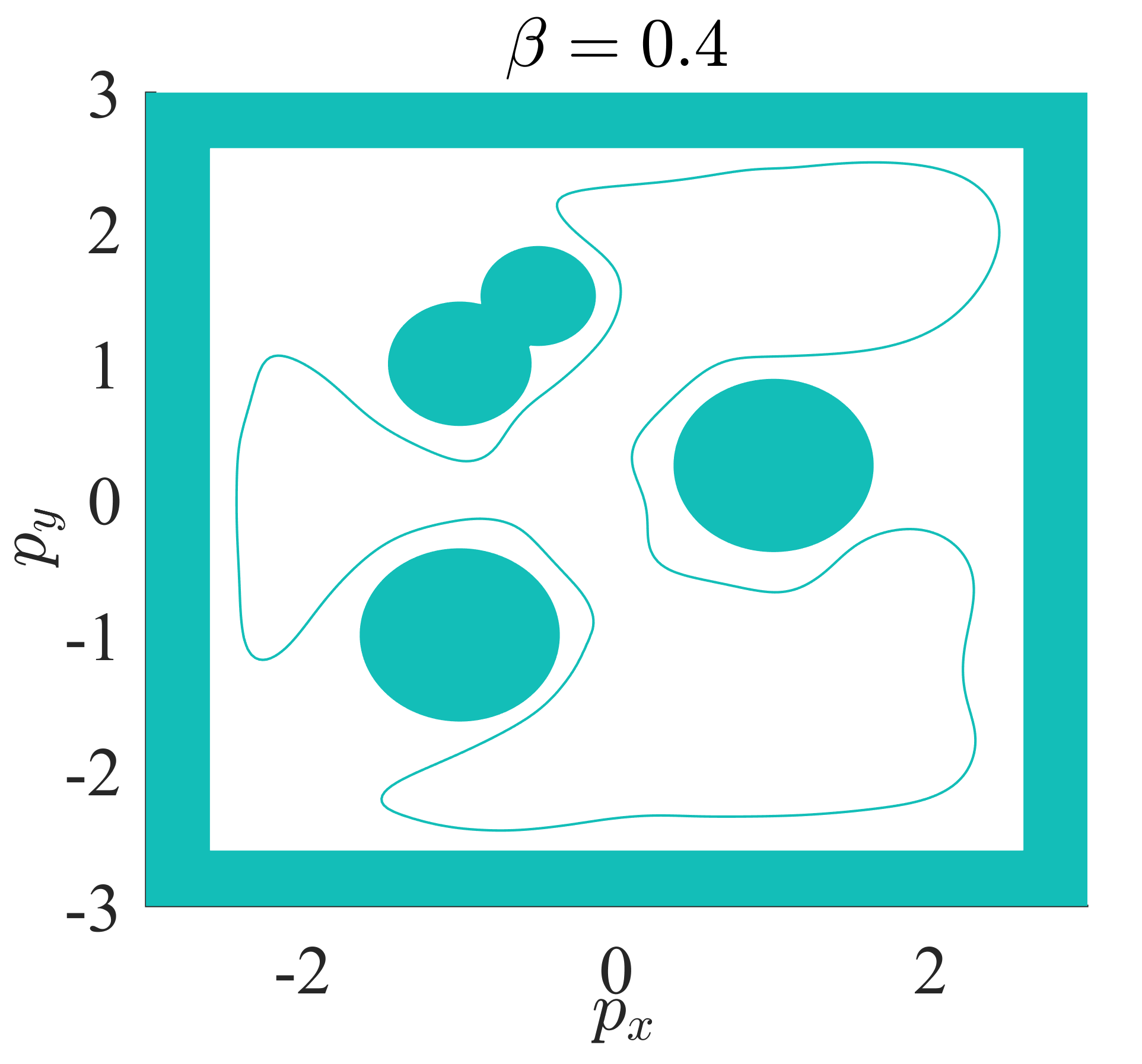}}
	\hfil
	\subfloat[]{\includegraphics[width=120pt,clip]{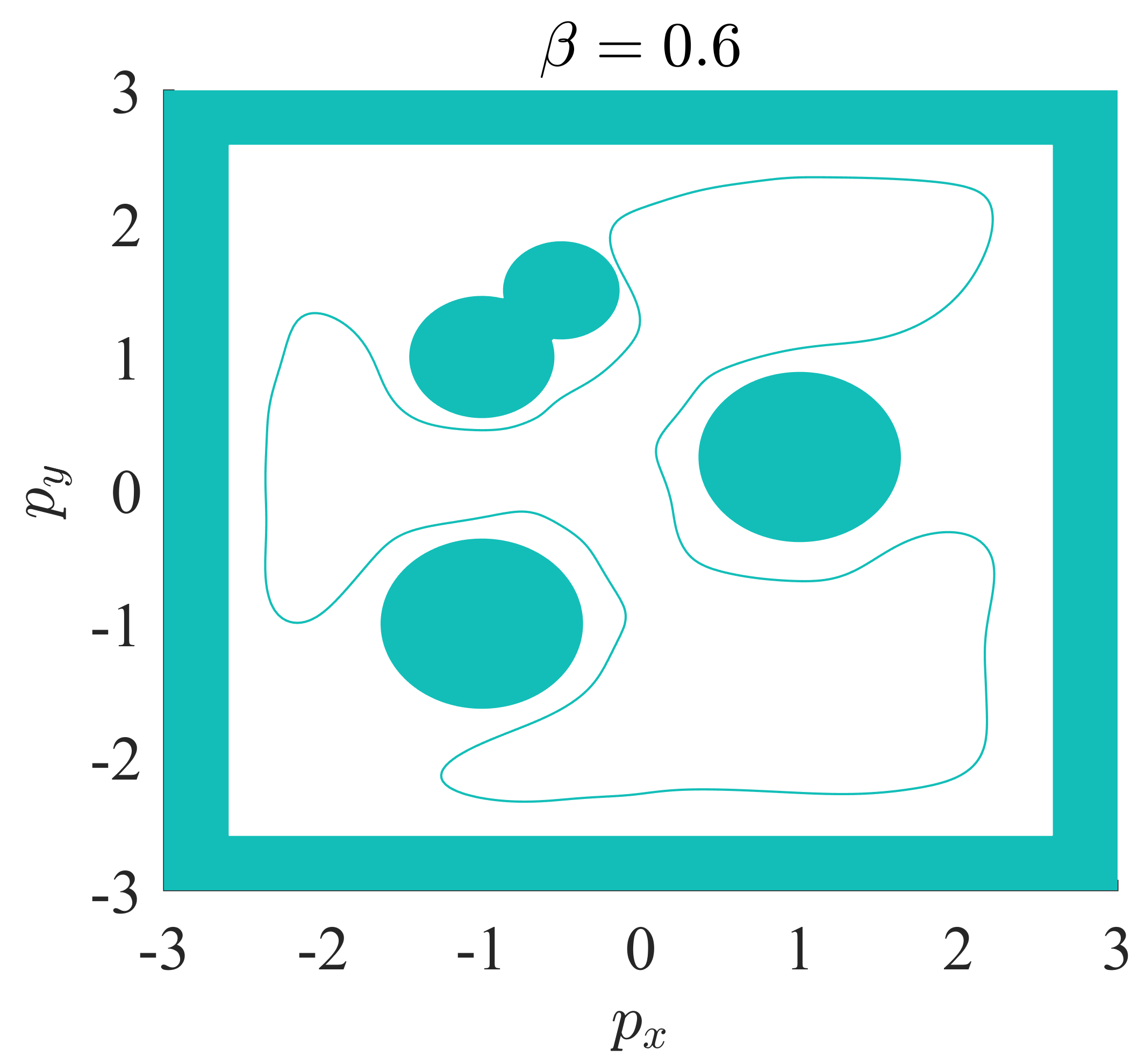}}
    \hfil
	\subfloat[]{\includegraphics[width=120pt,clip]{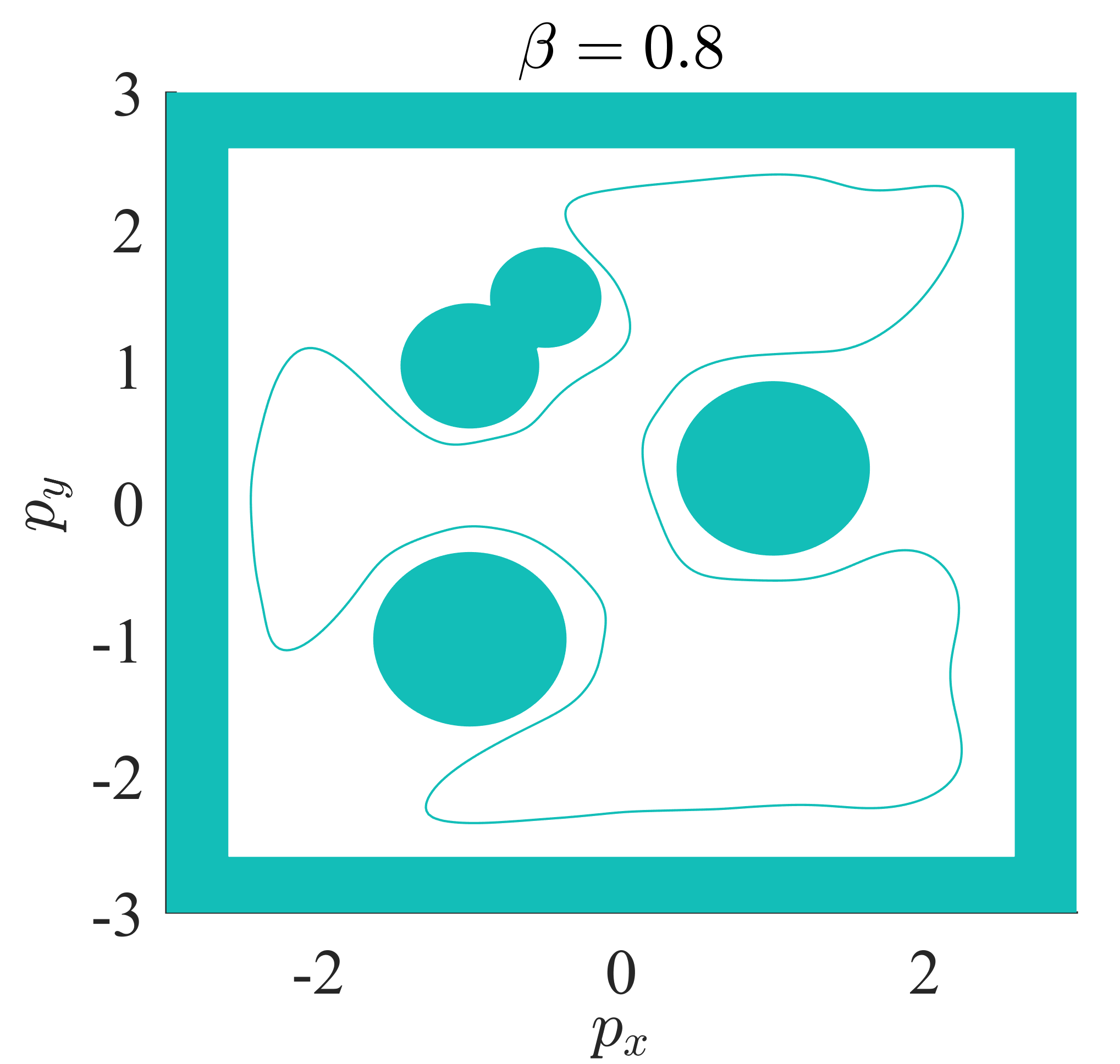}}
	\caption{Illustration of the safe set learned from the expert controller under different $\beta$.}
    \label{beta}
\end{figure}

\subsection{Closed-loop simulation} Hereafter, we evaluate the performance of the proposed optimization-based controller \eqref{DDfilter}.  We construct two reference control policies obtained using: (i) APF and (ii) deep deterministic policy gradient (DDPG) \cite{LillicrapHPHETS15}. To highlight the performance of the safety filter, we only apply the attractive potential field to get the reference controller (referred to as unsafe APF). Similarly, in the case of DDPG, we do not incorporate safety constraints during the training process (referred to as unsafe DDPG). For DDPG, the stage cost is designed as $g(x,u)=\max\{0,p^2_x + p^2_y-0.1^2\} $, with a discount factor $\gamma =0.99$. For comparison, we evaluate the filtered policies against: (i) the APF method incorporating both attractive and repulsive potential fields (referred to as the safe APF); (ii) the DDPG policy trained with a penalty on constraint violations in the cost function (referred to as safe DDPG) \cite{gupta2020policy}.

Fig. \ref{trajectory}(a) demonstrates the performance of the 3DSF learned using the expert-guided learning approach. When applying the unsafe policy, the trajectories exhibit some constraint violations. In contrast, the learned SACBF enables the vehicle to smoothly avoid obstacles, except for the initial condition $p_x = 2,\;p_y = 0.5$ (the rightmost black dot). To explain this constraint violation, we note that the expert controller (the safe APF) also violates the constraints when initialized from this state. This occurs because the APF becomes trapped in a local optimum located inside the obstacle in the right half-plane. Consequently, the 3DSF, which is learned based on this expert controller, recognizes this initial state as unsafe and therefore fails to refine the reference policy. The green curve in Figs. {trajectory}(a) and {trajectory}(c), representing the boundary of the safe set in the $p_x-p_y$ plane with $v=0$ and $\Psi = -\pi/2$ (corresponding to the speed and heading of the rightmost initial state), evidences this explanation by showing that this initial state is outside the safe set.”

In comparison, the 3DSF learned by RL refines the reference policy for all the initial states. The trajectory starting from the rightmost initial state in Fig. {trajectory}(b) demonstrates the superiority of the proposed 3DSF learned by RL. In particular, the safe APF policy makes the trajectory fail into a local optimum inside the right obstacle, while the unsafe APF policy with the 3DSF circumvents the local optimum and steers the system to the terminal set.

In Figs. {trajectory}(c-d), we show the trajectories of the vehicle controlled by the DDPG policies. It is found that although for some initial states safe DDPG successfully plans a safe trajectory, some trajectories fail to converge to the terminal set. This illustrates that it is hard to balance the task performance and constraint satisfaction by naively adding penalties to the stage cost during training. In contrast, the safety filter significantly enhances the safety performance of the unsafe DDPG policy while ensuring that the policy reaches the terminal set.
\begin{figure}
	\centering
	\subfloat[Dashed blue curve: unsafe APF; Solid blue curve: unsafe APF with 3DSF learned from the expert controller; Red curve: safe APF. ]{\includegraphics[width=120pt,clip]{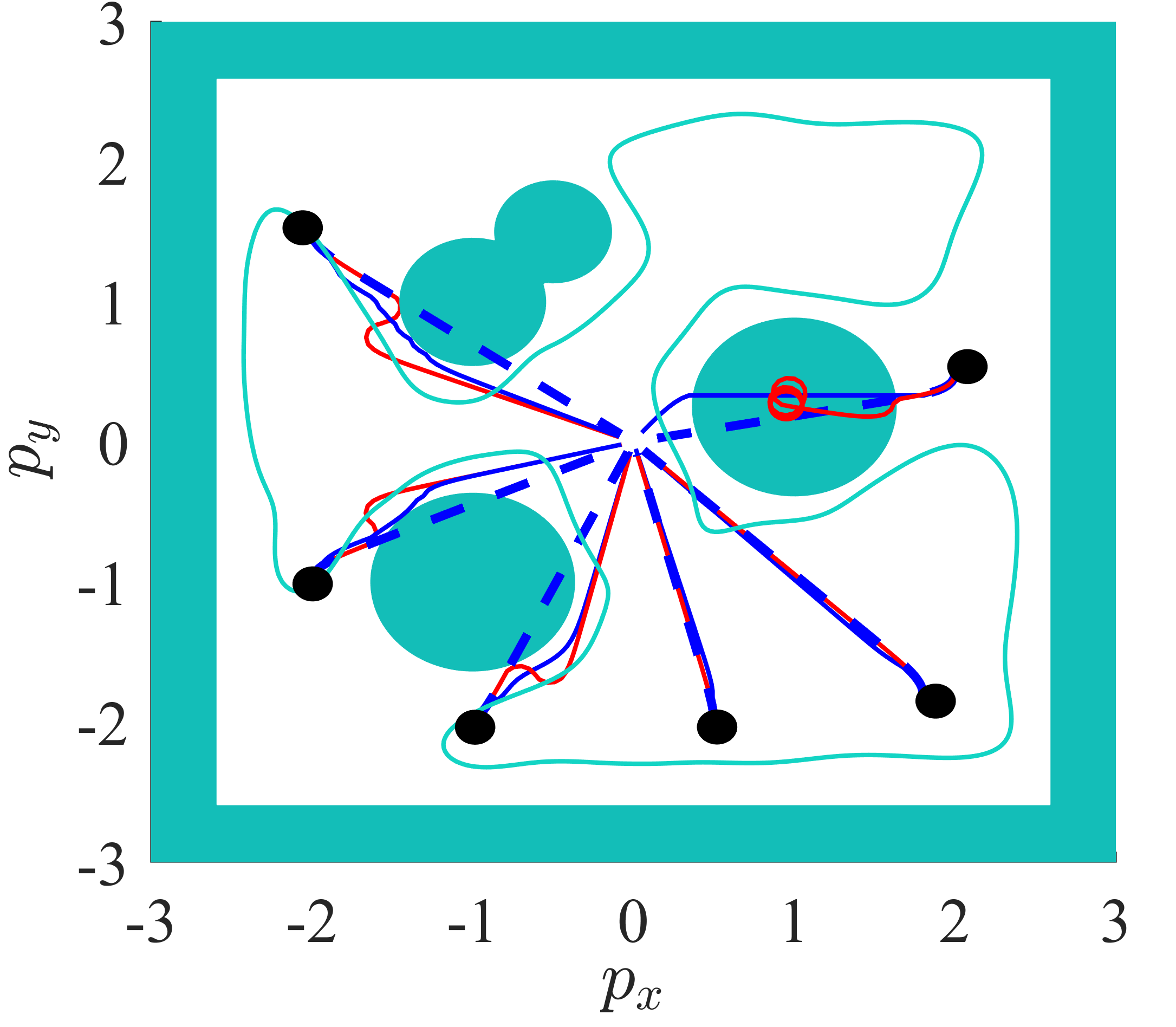}}
	\hfil
	\subfloat[Dashed blue curve: unsafe APF; Solid blue curve: unsafe APF with 3DSF learned by RL; Red curve: safe APF.]{\includegraphics[width=120pt,clip]{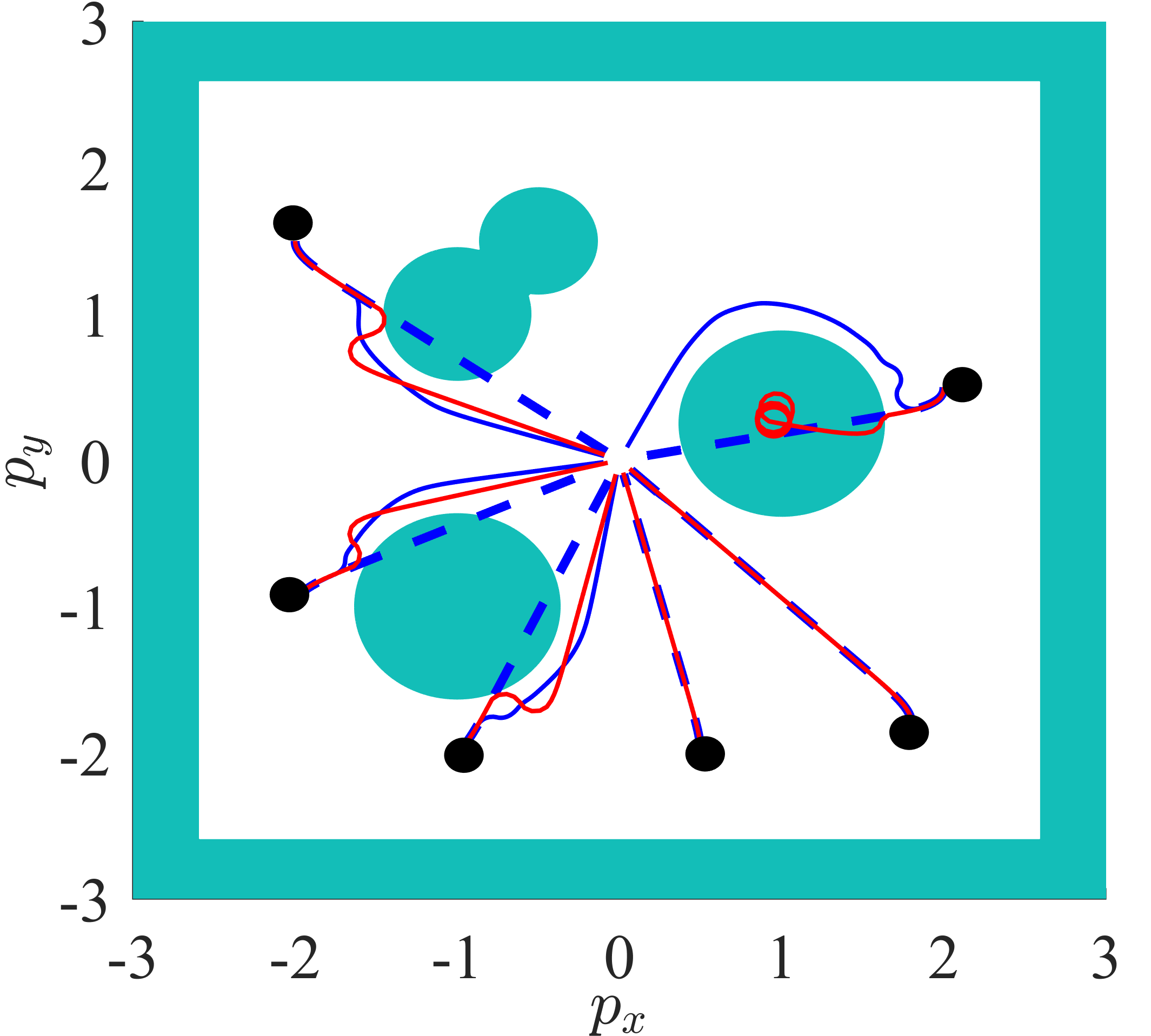}}
	\hfil\\
	\subfloat[Dashed blue curve: unsafe DDPG; Solid blue curve: unsafe DDPG with 3DSF learned from the expert controller; Red curve: safe DDPG.]{\includegraphics[width=120pt,clip]{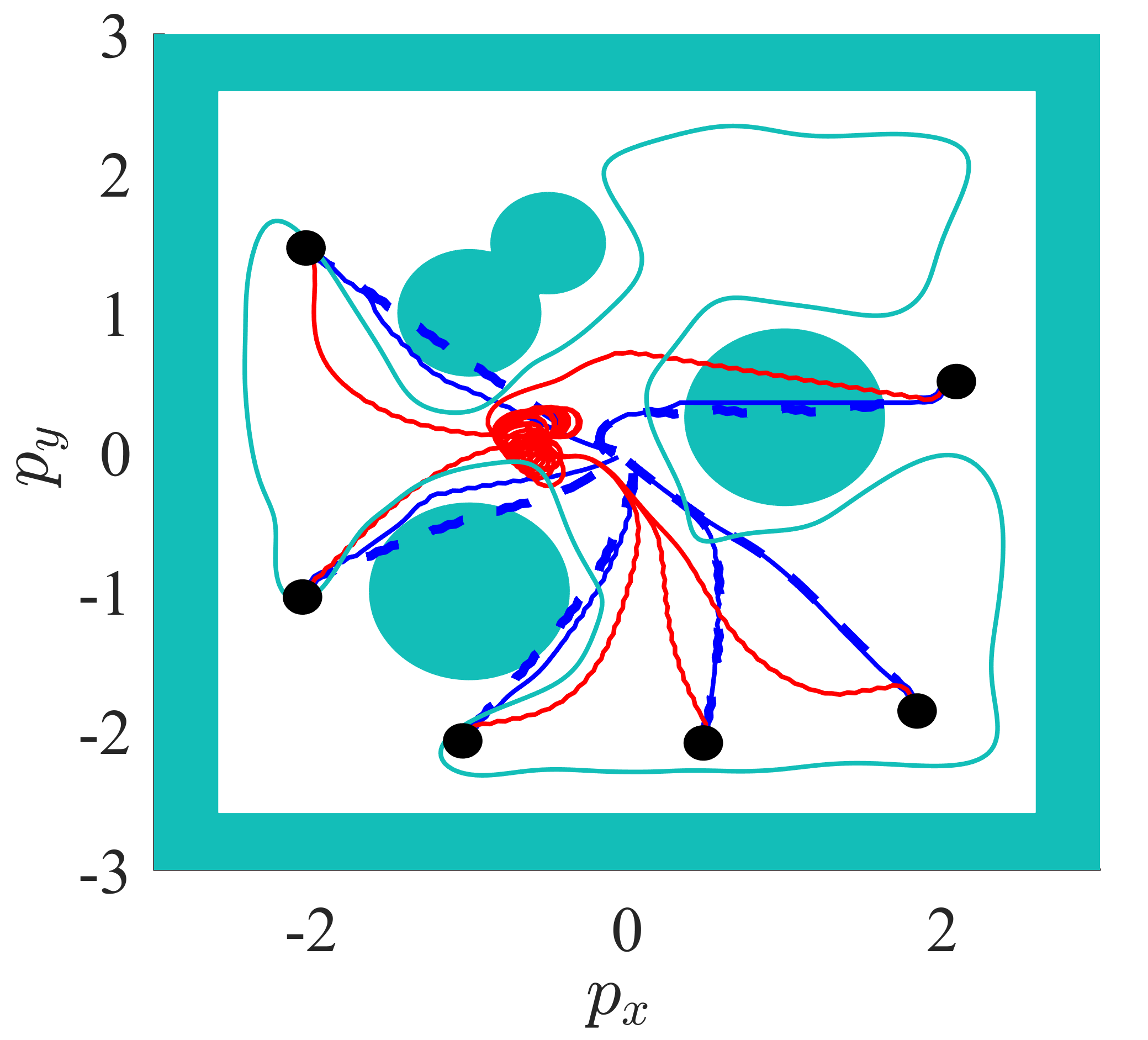}}
	\hfil
	\subfloat[Dashed blue curve: unsafe DDPG; Solid blue curve: unsafe DDPG with 3DSF learned by RL; Red curve: safe DDPG.]{\includegraphics[width=120pt,clip]{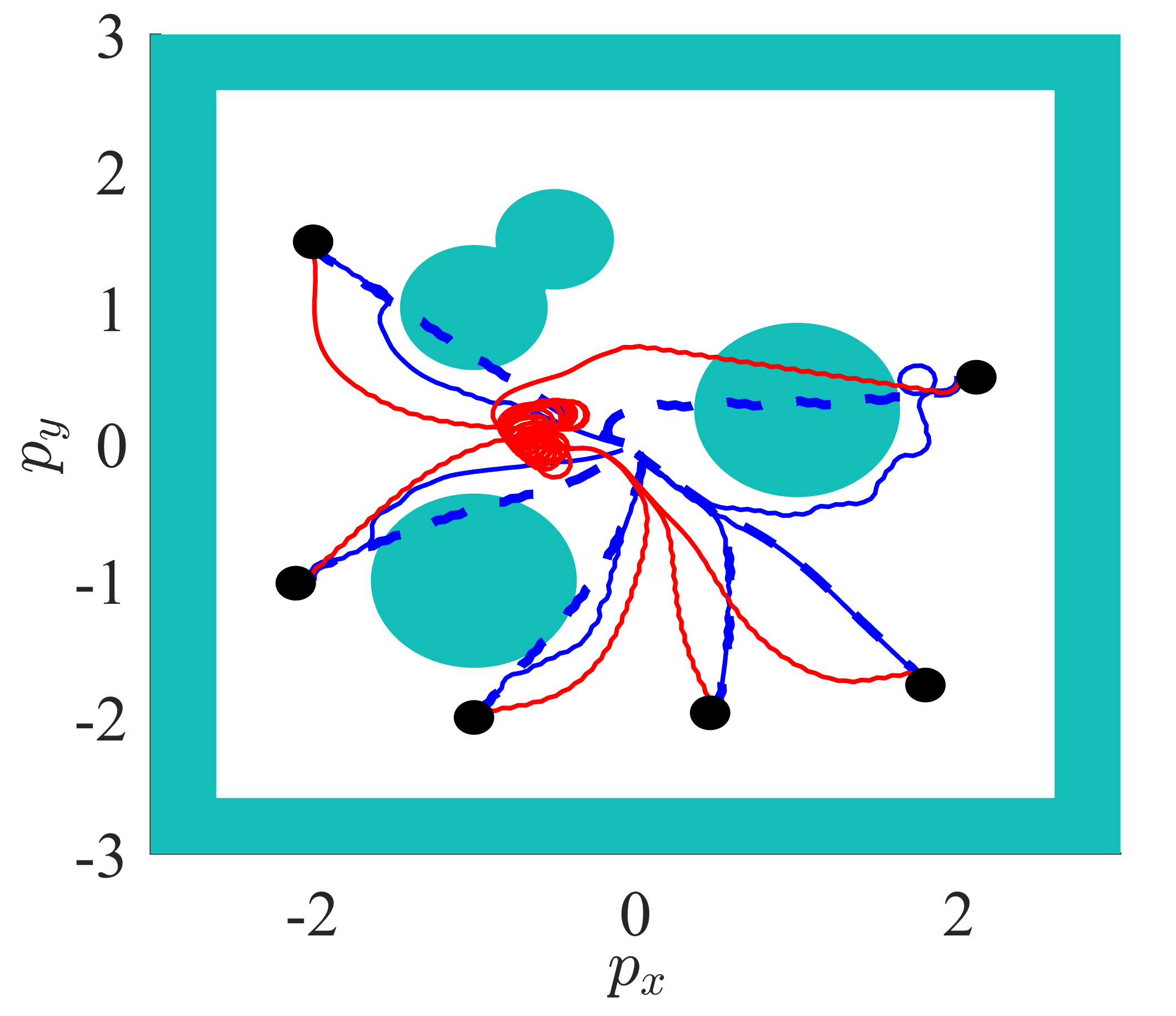}}
	\caption{Comparison of closed-loop vehicle trajectories under different controllers and different safety filters. Black dots represent initial states. Green curve: the boundary of the safe set learned from the expert
controller.}
    \label{trajectory}
\end{figure}
	
We have also implemented an MPC controller with a prediction horizon of 20 and a stage cost identical to that used in the DDPG controller. All state and input constraints are enforced as hard constraints, resulting in a constrained nonlinear optimization problem that is solved using active set. The MPC problem is warm-started using the shifted optimal solution from the previous time step. Fig. \ref{MPC_figure} illustrates the performance of the MPC controller, where the initial states are chosen to match those used for the other controllers in the simulation. As shown in the figure, all trajectories remain strictly within the feasible region. However, the trajectory starting from $p_x=2, p_y=0.5$ does not reach the terminal set. This is likely due to the non-convex nature of the MPC optimization problem, which may cause the solver to converge to a globally suboptimal control sequence.
\begin{figure}
	\centering
	\includegraphics[width=120pt,clip]{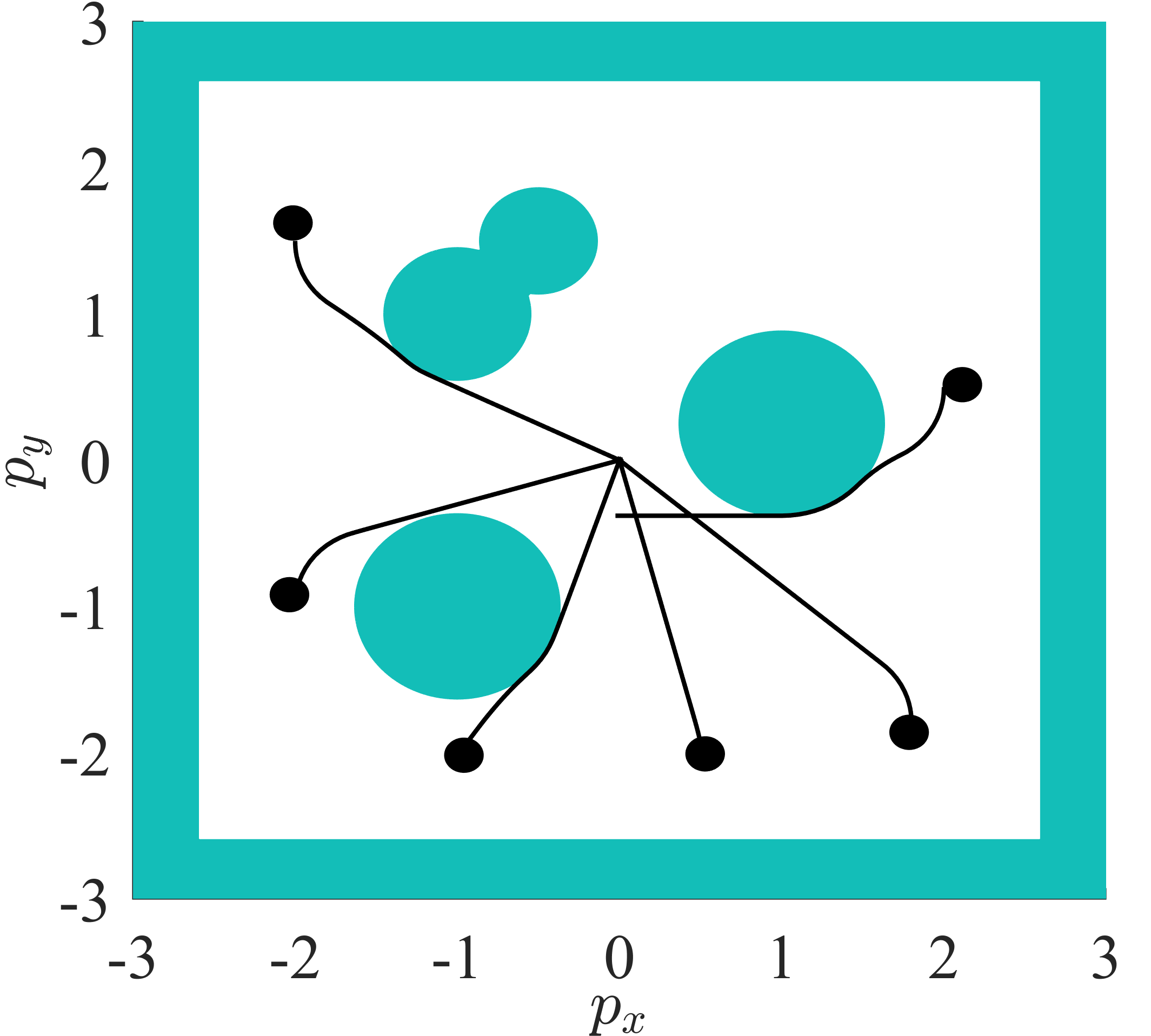}
	\caption{Closed-loop trajectories of MPC.}
    \label{MPC_figure}
\end{figure}

\subsection{Statistical evaluation} 

\begin{figure}
	\subfloat[Trajectories of the vehicle controlled by the unsafe APF policy filtered by the 3DSF.]{\includegraphics[width=120pt,clip]{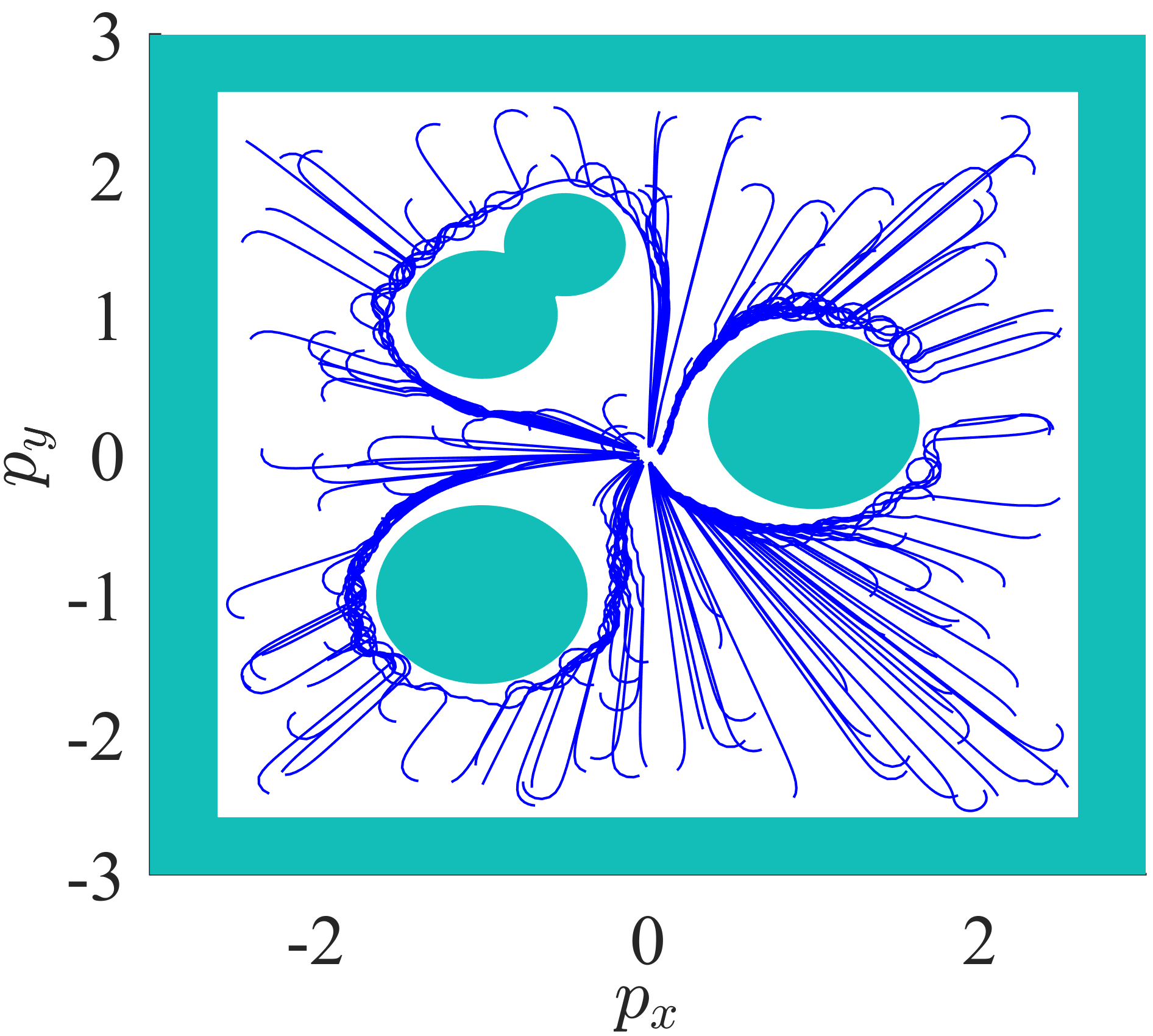}}
	\hfil
	\subfloat[Trajectories of the vehicle controlled by the unsafe APF policy filtered by the indirect safety filter.]{\includegraphics[width=120pt,clip]{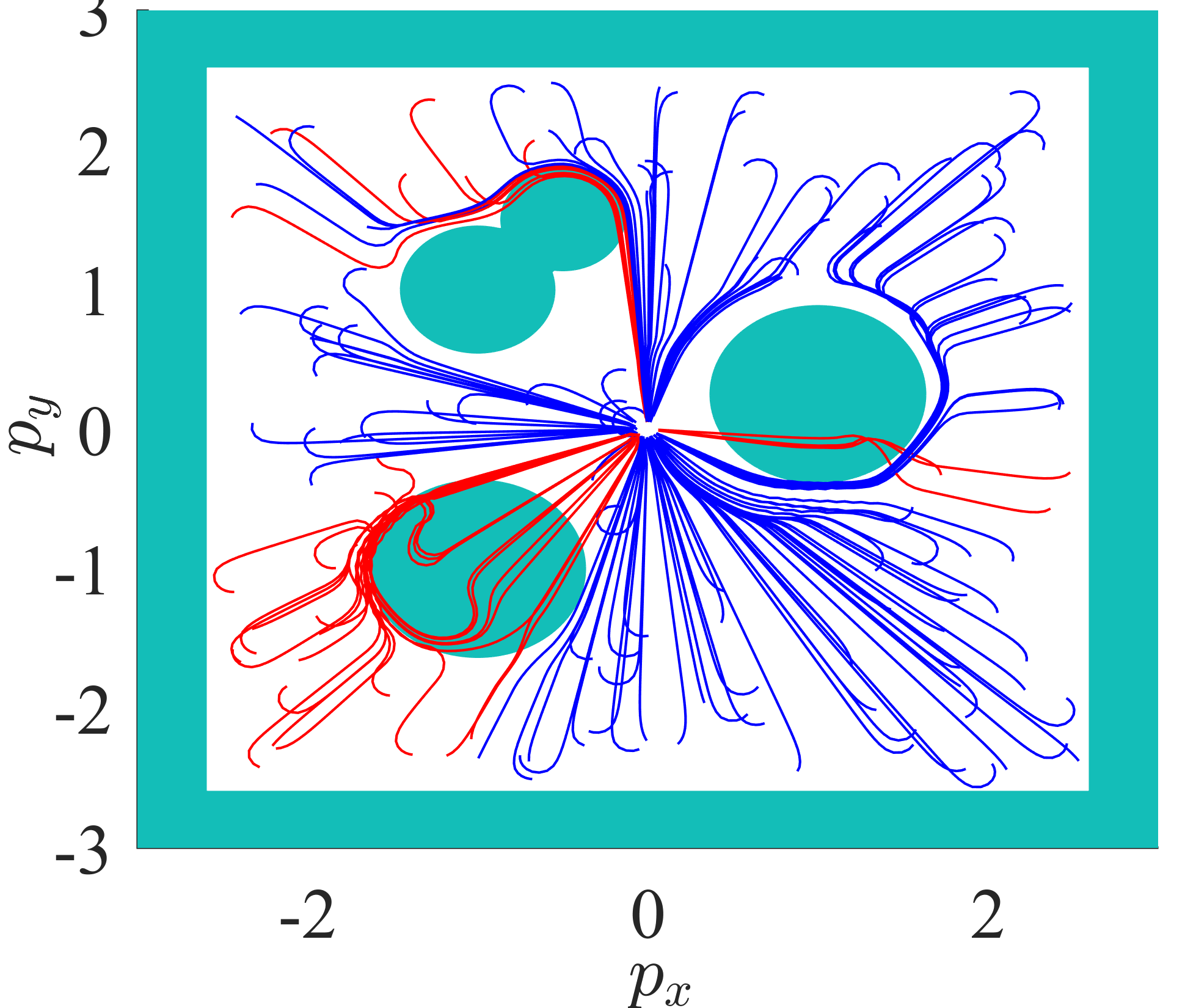}}
    \hfil
	\subfloat[Trajectories of the vehicle controlled by MPC.]{\includegraphics[width=120pt,clip]{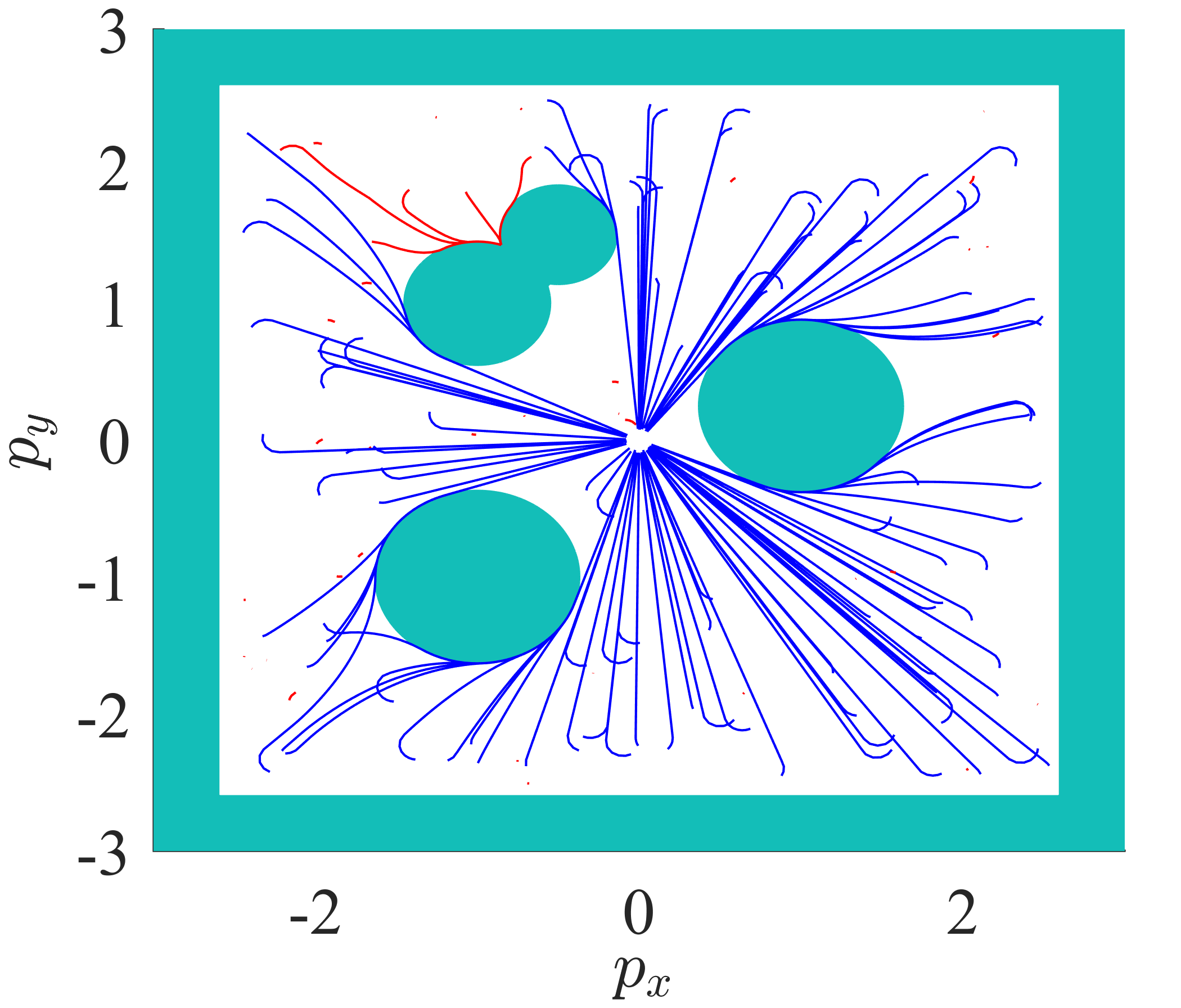}}
	\caption{Comparison of closed-loop vehicle trajectories under the 3DSF, the indirect safety filter, and MPC. }
    \label{trajectory_many}
\end{figure}

Finally, we statistically compare the proposed 3DSF against an indirect data-driven safety filter and MPC. In this subsection, the SACBF is learned by the RL approach because it has been illustrated in the previous subsection that the SACBF learned by the RL approach results in a larger safe set. To design the indirect data-driven safety filter, we intend to learn a standard CBF based on an approximate model. We first use a neural network with [128 128 128] “elu” hidden layers to identify the kinematic model. After 30 epochs of training, we get an approximate model with a sufficiently small ($2.38 \cdot 10^{-4}$) mean square error. Then, following the approaches of \cite{chen2024learning,fisac2018general,he2023state}, the approximate model is used as a prediction model in the reachability problem\footnote{The reachability problem \eqref{CBF generator} is defined over an infinite horizon. To make it solvable for each sampled state, we truncate the horizon at 20. While a longer horizon provides a better approximation of the maximal safe set, a horizon of 20 balances between sample complexity and safety performance.} \eqref{CBF generator}. We solve \eqref{CBF generator} with the initial state $x_0$ being each state sample obtained in Section \ref{section RL SACBF}. Then, we get the target value for the neural CBF, which is parameterized by a neural network with [128 128 32] “tansig” hidden layers. 

\hk{The MPC stage cost is designed as $x^TQx+u^TRu$ with weighting matrices $Q = \mathrm{diag}([1,\; 1, \;0,\; 0])$ and $R = \mathrm{diag}([0.001,\;0.001])$. The constraints of MPC are identical to the state and input constraints. Different prediction horizons (15, 20, and 25) are considered. The optimization problems involved in MPC and the proposed controller are nonlinear and solved using MATLAB’s \texttt{fmincon} function with the active-set algorithm. The initial guess is set to zero.}

\begin{table}
	\centering
	\caption{Comparison of safety rate, total cost, and average online CPU time for different control policies and safety filters.}
	\label{tab:policy_safety_filters}
	\begin{tabular}{@{}lcccc@{}}
		\toprule
		\textbf{Controller} & \textbf{Performance} & \textbf{No SF} & \textbf{Indirect SF} & \textbf{3DSF} \\ \midrule
		Unsafe APF              & 
		\begin{tabular}[c]{@{}c@{}}Safety Rate \\ Total Cost \\ CPU Time\end{tabular} & 
		\begin{tabular}[c]{@{}c@{}}58.16\% \\ 114.30 \\  0.0043 ms\end{tabular} & 
		\begin{tabular}[c]{@{}c@{}}79.43\% \\ 136.59 \\ 1.60 ms\end{tabular} & 
		\begin{tabular}[c]{@{}c@{}}100\% \\ 154.25\\  2.19 ms\end{tabular} \\ \midrule
		Safe APF           & 
		\begin{tabular}[c]{@{}c@{}}Safety Rate \\ Total Cost \\ CPU Time\end{tabular} & 
		\begin{tabular}[c]{@{}c@{}}81.56\% \\ 373.39\\ 0.0060 ms\end{tabular} & 
		\begin{tabular}[c]{@{}c@{}}80.85\% \\ 521.46 \\ 1.51 ms\end{tabular} & 
		\begin{tabular}[c]{@{}c@{}}100\% \\ 155.87 \\ 2.21 ms\end{tabular} \\ \midrule
		Unsafe DDPG          & 
		\begin{tabular}[c]{@{}c@{}}Safety Rate \\ Total Cost\\ CPU Time\end{tabular} & 
		\begin{tabular}[c]{@{}c@{}}70.92\% \\ 120.33 \\  0.39 ms\end{tabular} & 
		\begin{tabular}[c]{@{}c@{}}79.43\% \\ 141.22 \\ 2.21 ms\end{tabular} & 
		\begin{tabular}[c]{@{}c@{}}96.45\% \\ 187.54 \\ 2.49 ms\end{tabular} \\
		\midrule
		Safe DDPG       & 
		\begin{tabular}[c]{@{}c@{}}Safety Rate \\ Total Cost\\ CPU Time\end{tabular} & 
		\begin{tabular}[c]{@{}c@{}}82.27\% \\ 560.60\\  0.36 ms\end{tabular} & 
		\begin{tabular}[c]{@{}c@{}}89.36\% \\ 571.45 \\ 2.49 ms\end{tabular} & 
		\begin{tabular}[c]{@{}c@{}}100\% \\611.90\\ 2.61 ms\end{tabular} \\
		\midrule
		Constrained FQI       & 
		\begin{tabular}[c]{@{}c@{}}Safety Rate \\ Total Cost\\ CPU Time\end{tabular} & 
		-- & 
		-- & 
		\begin{tabular}[c]{@{}c@{}}100\% \\ 173.85\\ 4.95 ms\end{tabular} \\ 
        \midrule
		MPC      & 
		\begin{tabular}[c]{@{}c@{}}Horizon \\ Safety Rate \\ Total Cost\\ CPU Time\end{tabular} & 
		\begin{tabular}[c]{@{}c@{}}\hk{15} \\ \hk{95.58\%} \\ \hk{1101.80}\\ \hk{17.05 ms}\end{tabular} & 
		\begin{tabular}[c]{@{}c@{}}\hk{20} \\ \hk{100\%} \\ \hk{1196.01}\\ \hk{33.29 ms}\end{tabular} &  
		\begin{tabular}[c]{@{}c@{}}\hk{25} \\ \hk{100\%} \\ \hk{1214.80}\\ \hk{81.06 ms}\end{tabular}   \\
        \bottomrule
	\end{tabular}
\end{table}

We randomly sample 141 initial states that lie in the intersection of the safe set of the SACBF and the safe set of the CBF. Fig. \ref{trajectory_many} shows the trajectories controlled by the unsafe APF controller with the 3DSF and the indirect counterpart, respectively. Notably, the results demonstrate the absence of undesired equilibria, limit cycles, or unbounded trajectories.

As has been mentioned in Section \ref{pre}, \eqref{DDfilter} can be utilized either as a safety filter (3DSF) for a pre-obtained reference control policy or as a policy generator to determine suboptimal control inputs greedy to $Q_\theta$. Therefore, we apply the constrained FQI (Algorithm 1) to get the approximation $Q_\theta$ of the constrained optimal value function $\bar{Q}^*$ defined in \eqref{bellman2}. This makes \eqref{DDfilter} a greedy policy optimization problem.

The results in Table 1 provide a detailed comparison of the safety rate, total cost, and average CPU time for those trajectories starting from the sampled initial states. Each controller in Table 1 is evaluated under three different settings: (1) No SF (no safety filter): The controller is executed directly without any safety mechanism; (2) Indirect SF: The controller is safeguarded by the indirect safety filter; (3) 3DSF: The controller is filtered by the proposed 3DSF. For each combination, we evaluate (i): Safety Rate: The percentage of simulated trajectories that satisfy all constraints; (ii) Total Cost: The accumulated performance cost ($p^2_x + p^2_y$ over 1000 steps), which represents the task performance; (iii) CPU Time: The average online computation time required for computing a control input. The key findings are summarized below:
\begin{itemize}
	\item Both safety filters generally improve empirical safety rates, at the cost of varying changes in the accumulated task cost. When combined with safety filters, even unsafe policies (e.g., Unsafe APF and Unsafe DDPG) achieve safety rates comparable to their safe counterparts.
	\item Most importantly, the 3DSF significantly improves the safety rate compared to the indirect counterpart. The worse performance of the indirect counterpart is likely due to the superposition effect of the model error and the CBF learning error.
	\item Constrained FQI achieves a safety rate of 100\% and a lower total cost than the DDPG-based controllers with the 3DSF, suggesting that incorporating the SACBF constraint into the policy-optimization process can improve task performance.
    \item \hk{The average CPU time of MPC with prediction horizons 15 and 20 is lower than the sampling time. However, this timing reflects only the optimization time and does not include additional delays that may arise in a real implementation. In comparison, the CPU time of the proposed 3DSF is below 5 ms.}
\end{itemize}

\section{Conclusions and future work}
We proposed an optimization-based control framework that uses SACBF constraints to enforce safety in learning-based control without requiring an explicit model at runtime. SACBFs can be synthesized either from expert demonstrations or through RL, while the performance objective can be optimized separately. For the expert-guided approach, we established robustness and recursive-feasibility guarantees under bounded approximation errors. The vehicle case study demonstrated improved empirical constraint satisfaction compared with reward shaping and model-based safety filters, and substantially lower online computation time than MPC.

The limitations of our method include its sample complexity, the possibility of constraint violations during learning, and challenges in safety guarantees for the RL approach. \hk{Future work will focus on: (i) reducing sample complexity, (ii) developing new formulations of value-function-type SACBFs that guarantee both safety and learning convergence, and (iii) quantifying the suboptimality of the proposed controller.}

\bibliographystyle{IEEEtran}
\bibliography{references} 

@book{borrelli2017predictive,
	title={Predictive Control for Linear and Hybrid Systems},
	author={Borrelli, Francesco and Bemporad, Alberto and Morari, Manfred},
	year={2017},
	publisher={Cambridge University Press}
}

@article{lanza2025model,
  title={A model-free approach to control barrier functions using funnel control},
  author={Lanza, Lukas and K{\"o}hler, Johannes and Dennst{\"a}dt, Dario and Berger, Thomas and Worthmann, Karl},
  journal={IEEE Control Systems Letters},
  volume={9},
  pages={1183--1188},
  year={2025},
  publisher={IEEE}
}

@article{molnar2021model,
  title={Model-free safety-critical control for robotic systems},
  author={Molnar, Tamas G and Cosner, Ryan K and Singletary, Andrew W and Ubellacker, Wyatt and Ames, Aaron D},
  journal={IEEE Robotics and Automation Letters},
  volume={7},
  number={2},
  pages={944--951},
  year={2021},
  publisher={IEEE}
}

@inproceedings{chakrabarty2018data,
  title={Data-driven estimation of backward reachable and invariant sets for unmodeled systems via active learning},
  author={Chakrabarty, Ankush and Raghunathan, Arvind and Di Cairano, Stefano and Danielson, Claus},
  booktitle={2018 IEEE Conference on Decision and Control (CDC)},
  pages={372--377},
  year={2018},
}

@inproceedings{riedmiller2005neural,
	title={Neural fitted {Q} iteration--first experiences with a data efficient neural reinforcement learning method},
	author={Riedmiller, Martin},
	booktitle={European Conference on Machine Learning},
	pages={317--328},
	year={2005}
}

@book{busoniu2017reinforcement,
	title={Reinforcement Learning and Dynamic Programming Using Function Approximators},
	author={Busoniu, Lucian and Babuska, Robert and De Schutter, Bart and Ernst, Damien},
	year={2017},
	publisher={CRC Press}
}

@article{wabersich2021predictive,
	title={A predictive safety filter for learning-based control of constrained nonlinear dynamical systems},
	author={Wabersich, Kim Peter and Zeilinger, Melanie N},
	journal={Automatica},
	volume={129},
	pages={109597},
	year={2021},
	publisher={Elsevier}
}

@article{moreno2022predictive,
	title={Predictive Control with Learning-Based Terminal Costs Using Approximate Value Iteration},
	author={Moreno-Mora, Francisco and Beckenbach, Lukas and Streif, Stefan},
	journal={arXiv preprint arXiv:2212.00361},
	year={2022}
}

@article{hertneck2018learning,
	title={Learning an approximate model predictive controller with guarantees},
	author={Hertneck, Michael and K{\"o}hler, Johannes and Trimpe, Sebastian and Allg{\"o}wer, Frank},
	journal={IEEE Control Systems Letters},
	volume={2},
	number={3},
	pages={543--548},
	year={2018},
}

@article{gros2022learning,
	title={Learning for {MPC} with stability \& safety guarantees},
	author={Gros, Sebastien and Zanon, Mario},
	journal={Automatica},
	volume={146},
	pages={110598},
	year={2022},
	publisher={Elsevier}
}

@inproceedings{yu2022reachability,
	title={Reachability constrained reinforcement learning},
	author={Yu, Dongjie and Ma, Haitong and Li, Shengbo and Chen, Jianyu},
	booktitle={International Conference on Machine Learning},
	pages={25636--25655},
	year={2022},
	organization={PMLR}
}

@inproceedings{liu2020ipo,
	title={{IPO}: Interior-point policy optimization under constraints},
	author={Liu, Yongshuai and Ding, Jiaxin and Liu, Xin},
	booktitle={The AAAI Conference on Artificial Intelligence},
	volume={34},
	number={04},
	pages={4940--4947},
	year={2020}
}

@inproceedings{fisac2019bridging,
	title={Bridging {Hamilton-Jacobi} safety analysis and reinforcement learning},
	author={Fisac, Jaime F and Lugovoy, Neil F and Rubies-Royo, Vicen{\c{c}} and Ghosh, Shromona and Tomlin, Claire J},
	booktitle={2019 International Conference on Robotics and Automation (ICRA)},
	pages={8550--8556},
	year={2019},
	organization={IEEE}
}

@article{dalal2018safe,
	title={Safe exploration in continuous action spaces},
	author={Dalal, Gal and Dvijotham, Krishnamurthy and Vecerik, Matej and Hester, Todd and Paduraru, Cosmin and Tassa, Yuval},
	journal={arXiv preprint arXiv:1801.08757},
	year={2018}
}

@inproceedings{cheng2019end,
	title={End-to-end safe reinforcement learning through barrier functions for safety-critical continuous control tasks},
	author={Cheng, Richard and Orosz, G{\'a}bor and Murray, Richard M and Burdick, Joel W},
	booktitle={The AAAI Conference on Artificial Intelligence},
	volume={33},
	number={01},
	pages={3387--3395},
	year={2019}
}

@article{lin2023reinforcement,
	title={Reinforcement learning-based model predictive control for discrete-time systems},
	author={Lin, Min and Sun, Zhongqi and Xia, Yuanqing and Zhang, Jinhui},
	journal={IEEE Transactions on Neural Networks and Learning Systems},
	volume={35},
	number={3},
	pages={3312--3324},
	year={2023},
	publisher={IEEE}
}

@article{li2022using,
	title={Using stochastic programming to train neural network approximation of nonlinear {MPC} laws},
	author={Li, Yun and Hua, Kaixun and Cao, Yankai},
	journal={Automatica},
	volume={146},
	pages={110665},
	year={2022},
	publisher={Elsevier}
}

@inproceedings{robey2020learning,
	title={Learning control barrier functions from expert demonstrations},
	author={Robey, Alexander and Hu, Haimin and Lindemann, Lars and Zhang, Hanwen and Dimarogonas, Dimos V and Tu, Stephen and Matni, Nikolai},
	booktitle={2020 59th IEEE Conference on Decision and Control (CDC)},
	pages={3717--3724},
	year={2020},
}

@inproceedings{choi2021robust,
	title={Robust control barrier--value functions for safety-critical control},
	author={Choi, Jason J and Lee, Donggun and Sreenath, Koushil and Tomlin, Claire J and Herbert, Sylvia L},
	booktitle={2021 60th IEEE Conference on Decision and Control},
	pages={6814--6821},
	year={2021},
}

@article{massiani2022safe,
	title={Safe value functions},
	author={Massiani, Pierre-Fran{\c{c}}ois and Heim, Steve and Solowjow, Friedrich and Trimpe, Sebastian},
	journal={IEEE Transactions on Automatic Control},
	volume={68},
	number={5},
	pages={2743--2757},
	year={2022},
	publisher={IEEE}
}

@article{wabersich2022predictive,
title={Predictive control barrier functions: Enhanced safety mechanisms for learning-based control},
author={Wabersich, Kim P and Zeilinger, Melanie N},
journal={IEEE Transactions on Automatic Control},
volume={68},
number={5},
pages={2638--2651},
year={2022},
publisher={IEEE}
}

@inproceedings{bansal2021deepreach,
	title={Deepreach: A deep learning approach to high-dimensional reachability},
	author={Bansal, Somil and Tomlin, Claire J},
	booktitle={2021 IEEE International Conference on Robotics and Automation (ICRA)},
	pages={1817--1824},
	year={2021},
}

@article{dawson2023safe,
	title={Safe control with learned certificates: A survey of neural {Lyapunov}, barrier, and contraction methods for robotics and control},
	author={Dawson, Charles and Gao, Sicun and Fan, Chuchu},
	journal={IEEE Transactions on Robotics},
	volume={39},
	number={3},
	pages={1749--1767},
	year={2023},
	publisher={IEEE}
}

@article{ames2016control,
	title={Control barrier function based quadratic programs for safety critical systems},
	author={Ames, Aaron D and Xu, Xiangru and Grizzle, Jessy W and Tabuada, Paulo},
	journal={IEEE Transactions on Automatic Control},
	volume={62},
	number={8},
	pages={3861--3876},
	year={2016},
	publisher={IEEE}
}

@article{wabersich2023data,
	title={Data-driven safety filters: Hamilton-{Jacobi} reachability, control barrier functions, and predictive methods for uncertain systems},
	author={Wabersich, Kim P and Taylor, Andrew J and Choi, Jason J and Sreenath, Koushil and Tomlin, Claire J and Ames, Aaron D and Zeilinger, Melanie N},
	journal={IEEE Control Systems Magazine},
	volume={43},
	number={5},
	pages={137--177},
	year={2023},
	publisher={IEEE}
}

@article{lavanakul2024safety,
	title={Safety Filters for Black-Box Dynamical Systems by Learning Discriminating Hyperplanes},
	author={Lavanakul, Will and Choi, Jason J and Sreenath, Koushil and Tomlin, Claire J},
	journal={arXiv preprint arXiv:2402.05279},
	year={2024}
}

@article{he2024approximate,
	title={Approximate dynamic programming for constrained piecewise affine systems with stability and safety guarantees},
	author={He, Kanghui and Shi, Shengling and van den Boom, Ton and De Schutter, Bart},
	journal={IEEE Transactions on Systems, Man, and Cybernetics: Systems},
        volume={55},
	number={3},
	pages={1722--1734},
	year={2024},
	publisher={IEEE}
}

@article{isaly2024adaptive,
  title={Adaptive safety with a rise-based disturbance observer},
  author={Isaly, Axton and Patil, Omkar Sudhir and Sweatland, Hannah M and Sanfelice, Ricardo G and Dixon, Warren E},
  journal={IEEE Transactions on Automatic Control},
  volume={69},
  number={7},
  pages={4883--4890},
  year={2024},
  publisher={IEEE}
}

@article{kolathaya2018input,
	title={Input-to-state safety with control barrier functions},
	author={Kolathaya, Shishir and Ames, Aaron D},
	journal={IEEE Control Systems Letters},
	volume={3},
	number={1},
	pages={108--113},
	year={2018},
	publisher={IEEE}
}

@inproceedings{agrawal2017discrete,
	title={Discrete control barrier functions for safety-critical control of discrete systems with application to bipedal robot navigation.},
	author={Agrawal, Ayush and Sreenath, Koushil},
	booktitle={Robotics: Science and Systems},
	volume={13},
	pages={1--10},
	year={2017},
	organization={Cambridge, MA, USA}
}

@ARTICLE{he2023state,
  author={He, Kanghui and Shi, Shengling and van den Boom, Ton and De Schutter, Bart},
  journal={IEEE Transactions on Automatic Control}, 
  title={State-Action Control Barrier Functions: Imposing Safety on Learning-Based Control With Low Online Computational Costs}, 
  year={2026},
  volume={71},
  number={5},
  pages={3365-3371},
  doi={10.1109/TAC.2025.3636804}}

@article{watkins1992q,
	title={Q-learning},
	author={Watkins, Christopher JCH and Dayan, Peter},
	journal={Machine Learning},
	volume={8},
	pages={279--292},
	year={1992},
	publisher={Springer}
}

@inproceedings{warren1989global,
	title={Global path planning using artificial potential fields},
	author={Warren, Charles W},
	booktitle={1989 IEEE International Conference on Robotics and Automation},
	pages={316--317},
	year={1989},
	organization={IEEE Computer Society}
}

@article{ramezani2024generalization,
	title={On the generalization of stochastic gradient descent with momentum},
	author={Ramezani-Kebrya, Ali and Antonakopoulos, Kimon and Cevher, Volkan and Khisti, Ashish and Liang, Ben},
	journal={Journal of Machine Learning Research},
	volume={25},
	number={22},
	pages={1--56},
	year={2024}
}

@article{chen2024learning,
	title={Learning Performance-Oriented Control Barrier Functions Under Complex Safety Constraints and Limited Actuation},
	author={Chen, Shaoru and Fazlyab, Mahyar},
	journal={arXiv preprint arXiv:2401.05629},
	year={2024}
}

@article{fisac2018general,
	title={A general safety framework for learning-based control in uncertain robotic systems},
	author={Fisac, Jaime F and Akametalu, Anayo K and Zeilinger, Melanie N and Kaynama, Shahab and Gillula, Jeremy and Tomlin, Claire J},
	journal={IEEE Transactions on Automatic Control},
	volume={64},
	number={7},
	pages={2737--2752},
	year={2018},
	publisher={IEEE}
}

@inproceedings{LillicrapHPHETS15,
	author       = {Timothy P. Lillicrap and
	Jonathan J. Hunt and
	Alexander Pritzel and
	Nicolas Heess and
	Tom Erez and
	Yuval Tassa and
	David Silver and
	Daan Wierstra},
	title        = {Continuous control with deep reinforcement learning},
	booktitle    = {4th International Conference on Learning Representations, {ICLR} 2016},
	year         = {2016}}

@article{gupta2020policy,
	title={Policy-gradient and actor-critic based state representation learning for safe driving of autonomous vehicles},
	author={Gupta, Abhishek and Khwaja, Ahmed Shaharyar and Anpalagan, Alagan and Guan, Ling and Venkatesh, Bala},
	journal={Sensors},
	volume={20},
	number={21},
	pages={5991},
	year={2020},
	publisher={MDPI}
}

@article{mestres2025control,
	title={Control Barrier Function-Based Safety Filters: Characterization of Undesired Equilibria, Unbounded Trajectories, and Limit Cycles},
	author={Mestres, Pol and Chen, Yiting and Dall'anese, Emiliano and Cort{\'e}s, Jorge},
	journal={arXiv preprint arXiv:2501.09289},
	year={2025}
}

@inproceedings{he2020composite,
	title={Composite deep learning control for autonomous bicycles by using deep deterministic policy gradient},
	author={He, Kanghui and Dong, Chaoyang and Yan, An and Zheng, Qingyuan and Liang, Bin and Wang, Qing},
	booktitle={IECON 2020 The 46th Annual Conference of the IEEE Industrial Electronics Society},
	pages={2766--2773},
	year={2020},
}

@article{haydari2020deep,
	title={Deep reinforcement learning for intelligent transportation systems: A survey},
	author={Haydari, Ammar and Y{\i}lmaz, Yasin},
	journal={IEEE Transactions on Intelligent Transportation Systems},
	volume={23},
	number={1},
	pages={11--32},
	year={2020},
	publisher={IEEE}
}

@article{westenbroek2021combining,
	title={Combining model-based design and model-free policy optimization to learn safe, stabilizing controllers},
	author={Westenbroek, Tyler and Agrawal, Ayush and Casta{\~n}eda, Fernando and Sastry, S Shankar and Sreenath, Koushil},
	journal={IFAC-PapersOnLine},
	volume={54},
	number={5},
	pages={19--24},
	year={2021},
	publisher={Elsevier}
}

@inproceedings{taylor2020learning,
	title={Learning for safety-critical control with control barrier functions},
	author={Taylor, Andrew and Singletary, Andrew and Yue, Yisong and Ames, Aaron},
	booktitle={Learning for Dynamics and Control},
	pages={708--717},
	year={2020},
	organization={PMLR}
}

@article{dhiman2021control,
	title={Control barriers in bayesian learning of system dynamics},
	author={Dhiman, Vikas and Khojasteh, Mohammad Javad and Franceschetti, Massimo and Atanasov, Nikolay},
	journal={IEEE Transactions on Automatic Control},
	volume={68},
	number={1},
	pages={214--229},
	year={2021},
	publisher={IEEE}
}

@article{didier2024approximate,
	title={Approximate predictive control barrier function for discrete-time systems},
	author={Didier, Alexandre and Zeilinger, Melanie N},
	journal={arXiv preprint arXiv:2411.11610},
	year={2024}
}

@book{bertsekas2022abstract,
	title={Abstract Dynamic Programming},
	author={Bertsekas, Dimitri},
	year={2022},
	publisher={Athena Scientific}
}

@article{tan2023value,
	title={Value functions are control barrier functions: Verification of safe policies using control theory},
	author={Tan, Daniel CH and Acero, Fernando and McCarthy, Robert and Kanoulas, Dimitrios and Li, Zhibin},
	journal={arXiv preprint arXiv:2306.04026},
	year={2023}
}

@article{hoeffding1994probability,
	title={Probability inequalities for sums of bounded random variables},
	author={Hoeffding, Wassily},
	journal={The Collected Works of Wassily Hoeffding},
	pages={409--426},
	year={1994},
	publisher={Springer}
}

@article{oh2025safety,
  title={Safety with Agency: Human-Centered Safety Filter with Application to {AI}-Assisted Motorsports},
  author={Oh, Donggeon David and Lidard, Justin and Hu, Haimin and Sinhmar, Himani and Lazarski, Elle and Gopinath, Deepak and Sumner, Emily S and DeCastro, Jonathan A and Rosman, Guy and Leonard, Naomi Ehrich and others},
  journal={arXiv preprint arXiv:2504.11717},
  year={2025}
}

@inproceedings{so2024train,
  title={How to train your neural control barrier function: Learning safety filters for complex input-constrained systems},
  author={So, Oswin and Serlin, Zachary and Mann, Makai and Gonzales, Jake and Rutledge, Kwesi and Roy, Nicholas and Fan, Chuchu},
  booktitle={2024 IEEE International Conference on Robotics and Automation (ICRA)},
  pages={11532--11539},
  year={2024},
}

@article{zhang2025discrete,
  title={Discrete {GCBF} proximal policy optimization for multi-agent safe optimal control},
  author={Zhang, Songyuan and So, Oswin and Black, Mitchell and Fan, Chuchu},
  journal={arXiv preprint arXiv:2502.03640},
  year={2025}
}

@article{choi2023forward,
  title={A forward reachability perspective on robust control invariance and discount factors in reachability analysis},
  author={Choi, Jason J and Lee, Donggun and Li, Boyang and How, Jonathan P and Sreenath, Koushil and Herbert, Sylvia L and Tomlin, Claire J},
  journal={arXiv preprint arXiv:2310.17180},
  year={2023}
}

@inproceedings{hsu2023isaacs,
  title={{Isaacs}: Iterative soft adversarial actor-critic for safety},
  author={Hsu, Kai-Chieh and Nguyen, Duy Phuong and Fisac, Jaime Fernandez},
  booktitle={Learning for Dynamics and Control Conference},
  pages={90--103},
  year={2023},
  organization={PMLR}
}

@article{li2025certifiable,
  title={Certifiable Reachability Learning Using a New {Lipschitz} Continuous Value Function},
  author={Li, Jingqi and Lee, Donggun and Lee, Jaewon and Dong, Kris Shengjun and Sojoudi, Somayeh and Tomlin, Claire},
  journal={IEEE Robotics and Automation Letters},
  year={2025},
  publisher={IEEE}
}

@article{nejati2023data,
  title={Data-driven synthesis of safety controllers via multiple control barrier certificates},
  author={Nejati, Ameneh and Zamani, Majid},
  journal={IEEE Control Systems Letters},
  volume={7},
  pages={2497--2502},
  year={2023},
  publisher={IEEE}
}

@article{castaneda2025recursively,
  title={Recursively Feasible Probabilistic Safe Online Learning with Control Barrier Functions},
  author={Casta{\~n}eda, Fernando and Choi, Jason J and Jung, Wonsuhk and Zhang, Bike and Tomlin, Claire J and Sreenath, Koushil},
  journal={IEEE Open Journal of Control Systems},
  year={2025},
  publisher={IEEE}
}
\begin{IEEEbiography}[{\includegraphics[width=1in,height=1.25in,clip,keepaspectratio]{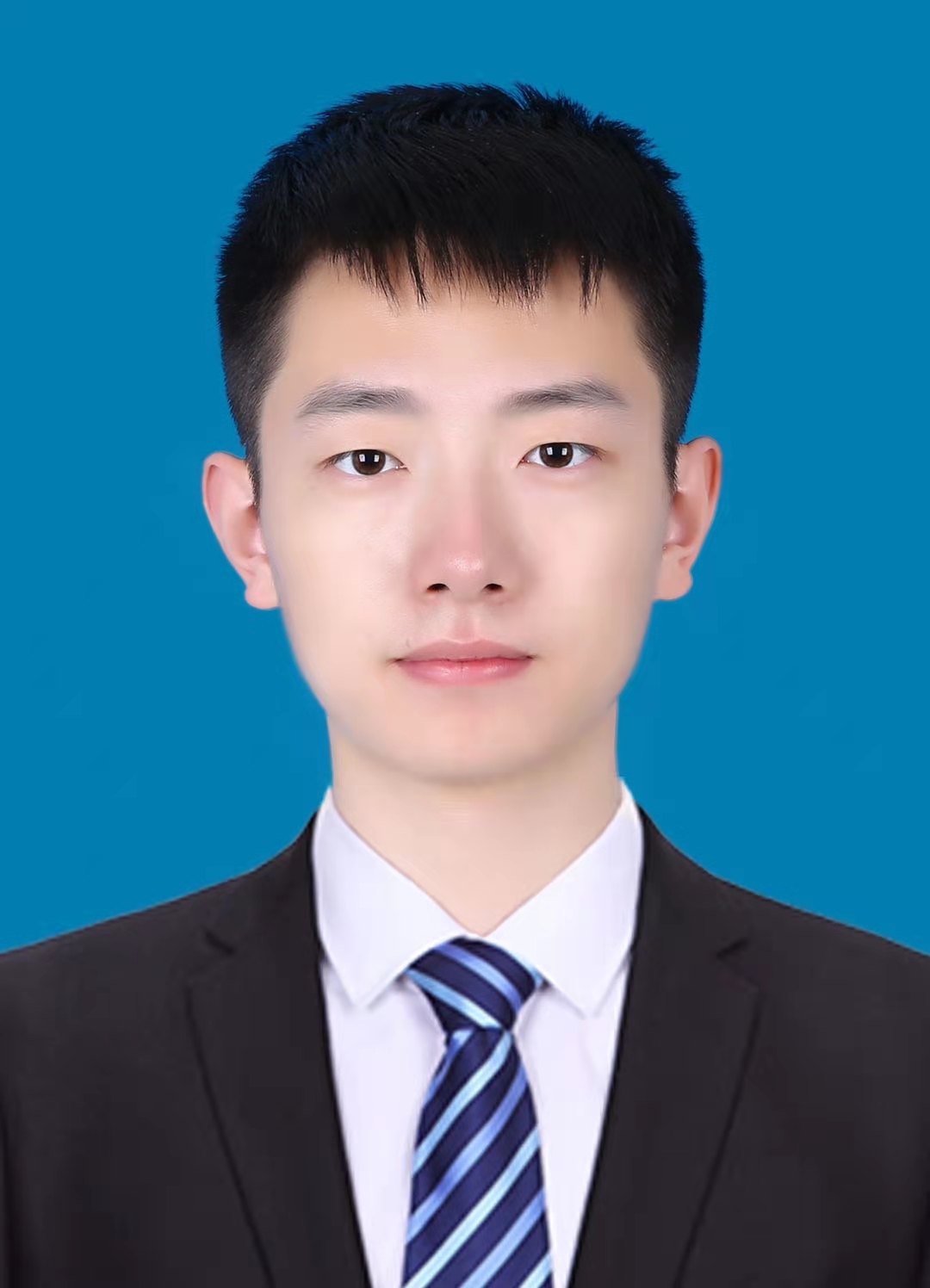}}]{Kanghui He} is a PhD candidate at the Delft Center for Systems and Control, Delft University of Technology, the Netherlands. He received the B.Sc. degree at the School of Mechanical Engineering and Automation, Beihang University, in 2018, and the M.Sc. degree (with outstanding graduation thesis and Chinese national scholarship) in the Department of Flight Dynamics and Control, Beihang University, in 2021. His research interests include learning-based control, model predictive control, hybrid systems, and their applications in mobile robots.
\end{IEEEbiography}
\begin{IEEEbiography}[{\includegraphics[width=1in,height=1.25in,clip,keepaspectratio]{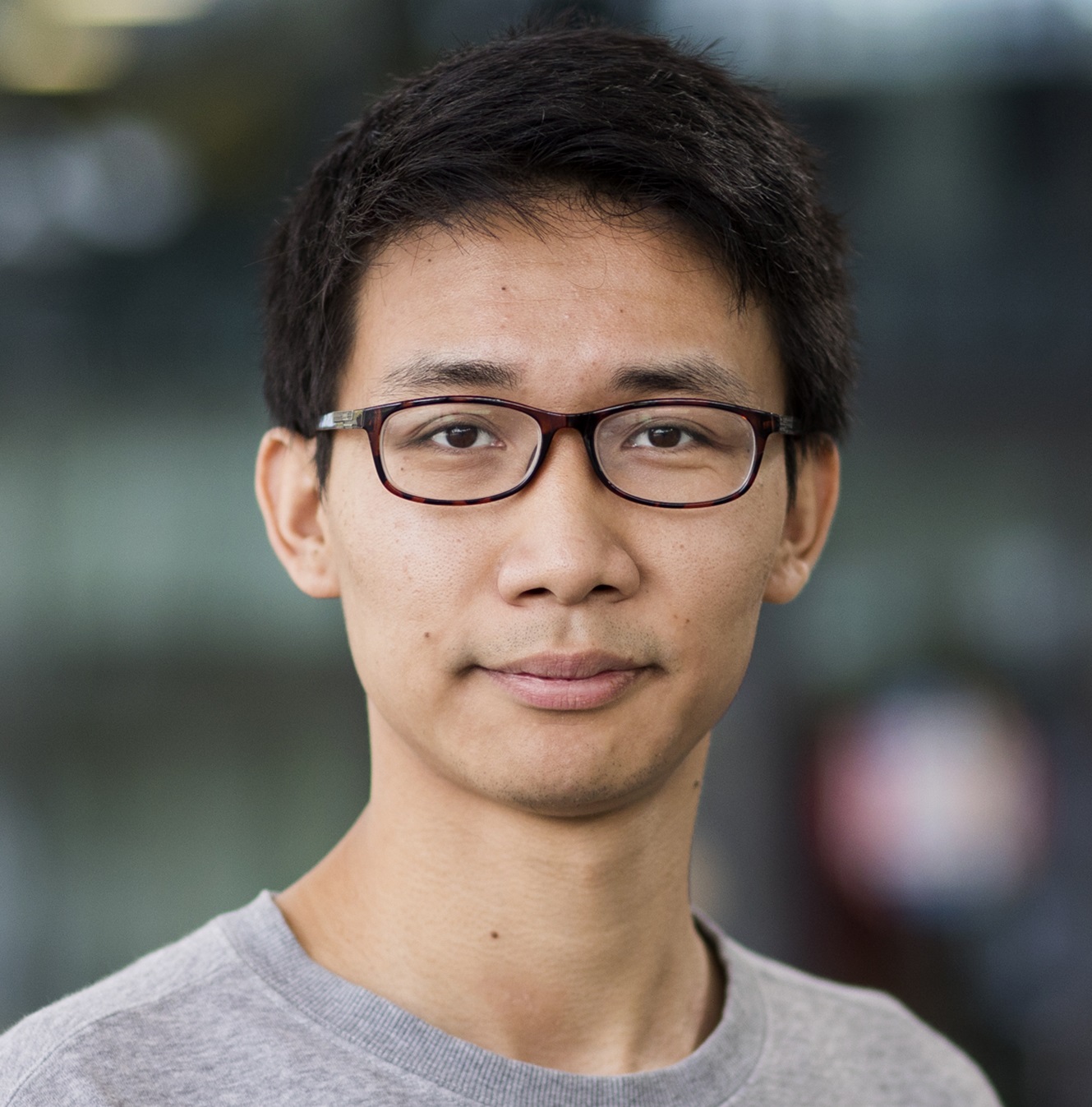}}]{Shengling Shi} is an Assistant Professor with the Delft Center for Systems and Control (DCSC), Delft University of Technology (TU Delft), Netherlands. He was a postdoctoral researcher at the Massachusetts Institute of Technology, USA, and at DCSC, TU Delft. He received the M.Sc. degree (with distinction) in 2017 and the Ph.D. degree in 2021 from the Eindhoven University of Technology, Netherlands. His research interests include system identification, optimal control, reinforcement learning, and their applications.
\end{IEEEbiography}
\begin{IEEEbiography}[{\includegraphics[width=1in,height=1.25in,clip,keepaspectratio]{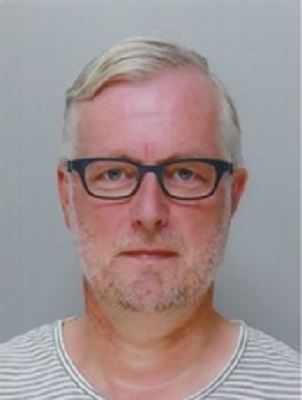}}]{Ton van den Boom} received his M.Sc. and Ph.D. degrees in Electrical Engineering from the Eindhoven University of Technology, The Netherlands, in 1988 and 1993, respectively. Currently, he is an associate professor at the Delft Center for Systems and Control (DCSC) department of Delft University of Technology. His research focus is mainly in modeling and control of discrete event and hybrid systems, in particular max-plus-linear systems, max–min-plus-scaling systems, and switching max-plus-linear systems (both stochastic and deterministic), with applications in manufacturing systems and transportation networks.
\end{IEEEbiography}
\begin{IEEEbiography}[{\includegraphics[width=1in,height=1.25in,clip,keepaspectratio]{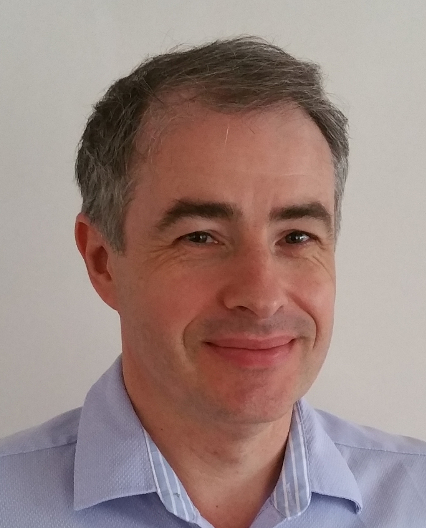}}]{Bart De Schutter} (Fellow, IEEE) received the PhD degree (summa cum laude) in applied sciences from KU Leuven, Belgium, in 1996. He is currently a Full Professor and Head of Department at the Delft Center for Systems and Control, Delft University of Technology, The Netherlands. His research interests include multi-level and multi-agent control, model predictive control, learning-based control, and control of hybrid systems, with applications in intelligent transportation systems and smart energy systems.
	
	Prof. De Schutter is a Senior Editor of the IEEE Transactions on Intelligent Transportation Systems.
\end{IEEEbiography}

\end{document}